\documentclass[10pt,conference,letterpaper]{IEEEtran}
\usepackage{color}
\usepackage{url}
\usepackage{graphicx, subfigure}
\usepackage{balance}  
\usepackage{algorithm}
\usepackage{algorithmic}
\usepackage{footnote}
\usepackage[utf8]{inputenc}
\usepackage{amsmath}
\usepackage{amssymb}
\usepackage{cite}
\usepackage{times,epsfig}
\begin{document}
\newtheorem{theorem}{Theorem}
\newtheorem{lemma}{Lemma}
\newtheorem{definition}{Definition}


\title{Faster Random Walks By Rewiring Online Social Networks On-The-Fly}


%
%
%
%

\author{
Zhuojie Zhou$^{1}$ \hspace{5mm}
Nan Zhang$^{1}$ \hspace{5mm}
Zhiguo Gong$^{2}$ \hspace{5mm}
Gautam Das$^{3}$ \\[3mm]
$^1$\affaddr{George Washington University}\\
$^2$\affaddr{University of Macau}\\
$^3$\affaddr{University of Texas at Arlington}}

\author{%
{Zhuojie Zhou{\small $^{1}$}, Nan Zhang{\small $^{2}$}, Zhiguo Gong{\small $^{3}$}, Gautam Das{\small $^{4}$}}%
\vspace{1.6mm}\\
\fontsize{10}{10}\selectfont\itshape
$~^{1,2}$Computer Science Department, George Washington University\\
\fontsize{9}{9}\selectfont\ttfamily\upshape
$~^{1}$rexzhou@gwu.edu\\
$~^{2}$nzhang10@gwu.edu%
\vspace{1.2mm}\\
\fontsize{10}{10}\selectfont\rmfamily\itshape
$~^{3}$Computer and Information Science Department, University of Macau\\
\fontsize{9}{9}\selectfont\ttfamily\upshape
$~^{2}$fstzgg@umac.mo
\vspace{1.2mm}\\
\fontsize{10}{10}\selectfont\rmfamily\itshape
$~^{4}$Computer Science Department, University of Texas at Arlington\\
\fontsize{9}{9}\selectfont\ttfamily\upshape
$~^{4}$gdas@uta.edu
}

\maketitle

\begin{abstract}
Many online social networks feature restrictive web interfaces which only allow the query of a user's local neighborhood through the interface. To enable analytics over such an online social network through its restrictive web interface, many recent efforts reuse the existing Markov Chain Monte Carlo methods such as random walks to sample the social network and support analytics based on the samples. The problem with such an approach, however, is the large amount of queries often required (i.e., a long ``mixing time'') for a random walk to reach a desired (stationary) sampling distribution.

In this paper, we consider a novel problem of enabling a faster random walk over online social networks by ``rewiring'' the social network on-the-fly. Specifically, we develop Modified TOpology (MTO)-Sampler which, by using only information exposed by the restrictive web interface, constructs a ``virtual'' overlay topology of the social network while performing a random walk, and ensures that the random walk follows the modified overlay topology rather than the original one. We show that MTO-Sampler not only provably enhances the efficiency of sampling, but also achieves significant savings on query cost over real-world online social networks such as Google Plus, Epinion etc.
\end{abstract}


%
%

\section{Introduction}

\subsection{Aggregate Estimation over Online Social Networks}

An online social network allows its users to publish contents and form connections with other users. To retrieve information from a social network, one generally needs to issue a {\em individual-user query} through the social network's web interface by specifying a user of interest, and the web interface returns the contents published by the user as well as a list of other users connected with the user\footnote{We currently focus on the undirected relationship between users.}. 

An online social network not only provides a platform for users to share information with their acquaintance, but also enables a third party to perform a wide variety of analytical applications over the social network - e.g., the analysis of rumor/news propagation, the mining of sentiment/opinion on certain subjects, and social media based market research. While some third parties, e.g., advertisers, may be able to negotiate contracts with the network owners to get access to the full underlying database, many third parties lack the resources to do so.  To enable these third-party analytical applications, one must be able to accurately estimate big-picture aggregates (e.g., the average age of users, the COUNT of user posts that contain a given word) over an online social network by issuing a small number of individual-user queries through the social network's web interface. We address this problem of third-party aggregate estimation in the paper.

\subsection{Existing Sampling Based Solutions and Their Problems}

An important challenge facing third-party aggregate estimation is the lack of cooperation from online social network providers. In particular, the information returned by each individual-user query is extremely limited - only containing information about the neighborhood of one user.  Furthermore, almost all large-scale online social networks enforce limits on the number of web requests one can issue (e.g., 600 open graph queries per 600 seconds for Facebook\footnote{https://developers.facebook.com/docs/best-practices/}, and 350 requests per hour for Twitter\footnote{https://dev.twitter.com/docs/rate-limiting}). As a result, it is practically impossible to crawl/download most or all data from an online social network before generating aggregate estimations.  There is also no available way for a third party to obtain the entire topology of the graph underlying the social network.

To address this challenge, a number of sampling techniques have been proposed for performing analytics over an online social network without the prerequisite of crawling \cite{Leskovec2006a, Jin, Katzir2011, Gjoka2010}. The objective of sampling is to randomly select elements (e.g., nodes/users or edges/relationships) from the online social network according to a pre-determined probability distribution, and then to generate aggregate estimations based on the retrieved samples. Since only individual local neighborhoods (i.e., a user and the set of its neighbors) - rather than the entire graph topology - can be retrieved from the social network's web interface, to the best of our knowledge, all existing sampling techniques without prior knowledge of all nodes/edges are built upon the idea of performing {\em random walks} over the graph which only require knowledge of the local neighborhoods visited by the random walks.

In literature, there are two popular random walk schemes: simple random walk and Metropolis Hastings random walk. {\em Simple random walk} (SRW) \cite{Lovasz1993} starts from an arbitrary user, repeatedly hops from one user to another by choosing uniformly at random from the former user's neighborhood, and stops after a number of steps to retrieve the last user as a sample. When the simple random walk is sufficiently long, the probability for each user to be sampled tends to reach a {\em stationary (probability) distribution} proportional to each user's degree (i.e., the number of users connected with the user). Thus, based on the retrieved samples and knowledge of such a stationary distribution, one can generate unbiased estimations of AVG aggregates (with or without selection conditions) over all users in the social network. If the total number of users in the social network is available\footnote{Which is the case for many real-world social networks whose providers publish the total number of users for advertising purposes.}, then COUNT and SUM aggregates can be answered without bias as well.

{\em Metropolis Hastings random walk} (MHRW) is a random walk achieving any distribution (typically uniform distribution) constructed by the famous MH algorithm. 
As an extension of MHRW, based on the knowledge of all the ids of a graph, 
\cite{Jin} suggests that we can conduct {\em random jump} (RJ), which jumps to any random vertex\footnote{It may need the global topology or the whole user id space for generate random vertex, thus not viable for all online social networks.} in the graph with a fixed probability in each step when it carries on the MHRW. Although MHRW can yield asymptotically uniform samples, which requires no additional processing for subsequent analysis, it is slower than SRW almost for all practical measurements of convergence, such as degree distribution distance, KS distance and mean degree error. According to \cite{Gjoka2010} and \cite{Lee2012}, SRW is 1.5-8 times faster than MHRW. Thus we set the baseline as SRW, while we also include MHRW in the experimental section.


A critical problem of existing sampling techniques, however, is the large number of individual-user queries (i.e., web requests) they require for retrieving each sample. Consider the above-described simple random walk as an example. In order to reach the stationary distribution (and thereby an accurate aggregate estimation), one may have to issue a large number of queries as a ``burn-in'' period of the random walk. Traditional studies on graph theory found that the length of such a burn-in period is determined by the graph {\em conductance} - an intrinsic property of the graph topology (formally defined in Section~\ref{preliminaries}). In particular, the smaller the conductance is, the longer the burn-in period will be (i.e., the more individual-user queries will be required by sampling).

Unfortunately, a recent study \cite{Mohaisena} on real-world social networks such as Facebook, Livejournal, etc.~found the conductance of their graphs to be substantially lower than expected. As a result, a random walk on these social networks often requires a large number of individual-user queries - e.g., approximately 500 to 1500 single random walk length for a real-world social network Livejournal of one million nodes to achieve acceptable variance distance \cite{Mohaisena}. One can see that, in order to retrieve enough samples to reach an accurate aggregate estimation, the existing sampling techniques may require a very large number of individual-user queries.

\subsection{Outline of Technical Results}

In this paper, we consider a novel problem of how to significantly increase the conductance of a social network graph by modifying the graph topology on-the-fly (during the third-party random walk process).  In the following, we shall first explain what we mean by on-the-fly topology modification, and then describe the rationale behind our main ideas for topology modification.

First, by topology modification we do {\em not} actually modify the original topology of the social network graph - indeed, no third party other than the social network provider has the ability to do so. What we modify is the topology of an {\em overlay graph} on which we perform the random walks. Fig~\ref{doublelayer} depicts an example: if we can decide that not considering a particular edge in the random walk process can make the burn-in period shorter (i.e., increase the conductance), then we are essentially performing random walks over an overlay graph on which this edge is removed. By doing so, we can achieve same accurate aggregate estimation with lower query cost.
One can see that, with traditional random walk techniques, the overlay graph is exactly the same as the original social network graph. Our objective here is to manipulate edges in the overlay graph so as to maximize the graph conductance.

It is important to note that the technical challenge here is not how edge manipulations can boost graph conductance - a simple method to reach theoretical maximum on conductance is to repeatedly insert edges to the graph until it becomes a complete graph. This requires the knowledge of all nodes in the social network, which a third-party does not have. The key challenge here is how to perform edge manipulations only based on the knowledge of local neighborhoods that a random walk has passed by, and yet increases the conductance of the entire graph in a significant manner. In the following, we provide an intuitive explanation of our main ideas for topology modification.

To understand the main ideas, we first introduce the concepts of {\em cross-cutting} and {\em non-cross-cutting} edges intuitively with an example in Fig~\ref{doublelayer} (we shall formally define these concepts in Section~\ref{preliminaries}). Generally speaking, if we consider a social network graph consisting of multiple densely connected components (e.g., $S$ and $\bar{S}$ in Fig~\ref{doublelayer}), then the edges connecting them are likely to be cross-cutting edges, while edges inside each densely connected component are likely non-cross-cutting ones. A key intuition here is that the more cross-cutting edges and/or the fewer non-cross-cutting edges a graph has, the higher its conductance is. For example, Graph $G$ in Fig~\ref{doublelayer} has a low conductance (i.e., high burn-in period) as a random walk is likely to get ``stucked'' in one of the two dense components which are difficult to escape, given that there is only one cross-cutting edge $(u, v)$. On the other hand, with far fewer non-cross-cutting edges and a few additional cross-cutting edges, $G^*$ has a much higher conductance as it is much easier now for a random walk to move from one component to the other.

\begin{figure}
    \centering
    \includegraphics[width=3.6in]{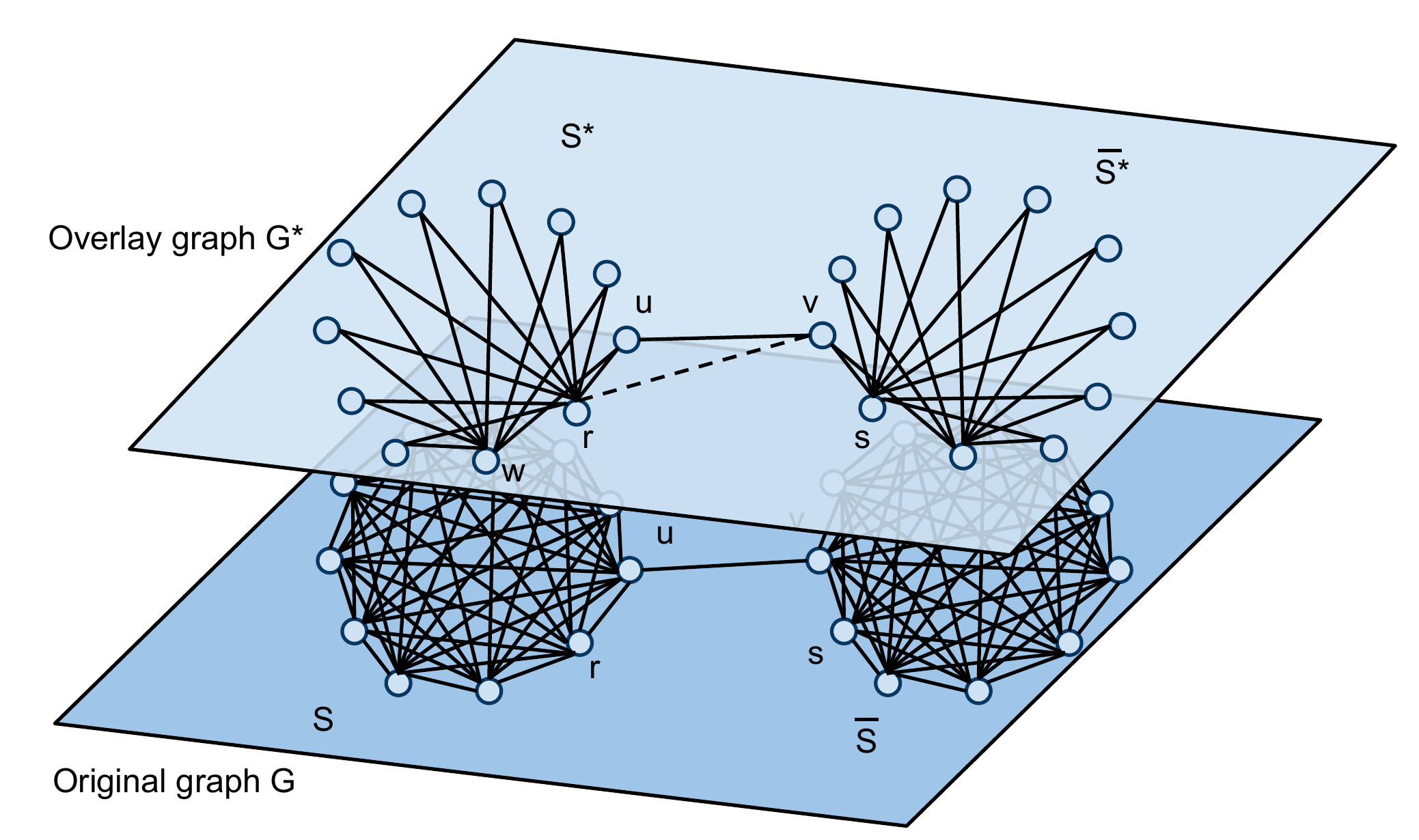}
    \caption{A concept of a random walk on the topologically modified overlay graph.}
    \label{doublelayer}
\end{figure}

With the concepts of cross-cutting and non-cross-cutting edges, we develop Modify TOpology Sampler (MTO-Sampler), a topology manipulation technique which first determines\footnote{Note that, as we shall prove in section \ref{sec:identification_challenge}, it is impossible to assert deterministically that an edge is cross-cutting. Nonetheless, it is possible to assert deterministically that an edge is non-cross-cutting. Thus, our algorithm has two possible outputs: non-cross-cutting or uncertain. We shall show in the paper that it outputs non-cross-cutting for a large number of (non-cross-cutting) edges in real-world social networks.} whether a given edge in the graph is a cross-cutting edge based solely upon knowledge of the local neighborhood topology, and then removes the edge if it is non-cross-cutting. MTO-Sampler may also ``move'' an edge by changing a node connected to the edge if it is determined that, by doing so, the new edge is more likely to be a cross-cutting edge. We shall show in the paper that MTO-Sampler is capable of significantly improving the efficiency of random walks: For the example in Fig~\ref{doublelayer}, MTO-Sampler is capable of reducing the mixing time (i.e., query cost of a random walk) by 97\%. We also demonstrate through experimental results the significant improvement of efficiency achieved by MTO-Sampler for real-world social networks such as Epinions, Google Plus, etc.

The main contributions of our approach include:
\begin{itemize}
\item (Problem Novelty) We consider a novel problem of modifying the graph topology on-the-fly (during the random walk process) for the efficient third-party sampling of online social networks.
\item (Solution Novelty) We develop MTO-Sampler which determines whether an edge is (non-)cross-cutting based solely upon local neighborhood knowledge retrieved by the random walk, and then manipulates the graph topology to significantly improve sampling efficiency.
\item Our contributions also include extensive theoretical analysis (on various social network models) and experimental evaluation on synthetic and real-world social networks as well as online at Google+ which demonstrate the superiority of our MTO-Sampler over the traditional sampling techniques.
\end{itemize}

\section{Preliminaries}
\label{preliminaries}

\subsection{Model of Online Social Networks}

In this paper, we consider an online social network with an interface that allows input queries of the form
\begin{center}
$q(v)$: SELECT * FROM D WHERE USER-ID = $v$,
\end{center}
and responds with the information about user $v$ (e.g., user name, self-description, user-published contents) as well as the list of all other users connected with $v$ (e.g., $v$'s friends in the network). This is a model followed by many online social networks - e.g., Google Plus, Facebook, etc - with the interface provided as either an end-user-friendly web page or a developer-specific API call.

Consider the social-network topology as an undirected graph $G(V, E)$, where each node in $V$ is corresponding to a user in the social network\footnote{Note that without introducing ambiguity, we use ``node'' and ``social network user'' interchangeably in this paper.}, and each edge in $E$ represents the connection between two users. One can see that the answer to query $q(v)$ ($v \in V$) is a set of nodes $N(v) \subseteq V$, such that $\forall u \in N(v)$, there is an edge $e: (u, v) \in E$. We henceforth refer to $N(v)$ as the {\em neighborhood} of $v$. We use $k_v$ to denote the {\em degree} of $v$ - i.e., $k_v = |N(v)|$. For abbreviation, we also write $e:(u,v)$ as $e_{uv}$.

\vspace{1mm}
\noindent \framebox[\columnwidth]{\parbox{0.9\columnwidth}{{\bf Running Example}: We shall use, throughout this paper, the 22-node, 111-edge, barbell graph shown (as the original graph $G$) in Fig~\ref{doublelayer} as a running example.}}
\vspace{1mm}


\subsection{Performance Measures for Sampling}
\label{pf_measure}
In the following, we shall discuss two key objectives for sampling: (1) minimizing bias - such that the retrieved samples can be used to accurately estimate aggregate query answers, and (2) reducing the number of queries required for sampling - given the stringent requirement often put in place by real-world social networks on the number of queries one can issue per day.

\vspace{1mm}\noindent{\em Bias:} In general, sampling bias is the ``distance'' between the target (i.e., ideal) distribution of samples and the actual sampling distribution - i.e., the probability for each tuple to be retrieved as a sample. We shall further discuss a concrete bias measure in the next subsection and an experimental measure in Section~\ref{sec:eps}.

\vspace{1mm}\noindent{\em Query Cost:} To this end, we consider the number of {\em unique} queries one has to issue for the sampling process, as any duplicate query can be answered from local cache without consuming the query limit enforced by the social network provider.

\subsection{Random Walk}

A random walk is a Markov Chain Monte Carlo (MCMC) method which takes successive random steps on the above-described graph $G$ according to a {\em transition matrix} $P=(p_{uv}), u,v\in V$, where $p_{uv}$ represents the probability for the random walk to transit from  node $u$ to $v$. The premise here is that, after performing a random walk for a sufficient number of steps, the probability distribution for the walk to land on each node in $G$ converges to a {\em stationary distribution} $\pi$ which then becomes the sampling distribution\footnote{That is, if we take the end node as a sample}. There are many different types of random walks, corresponding to the different designs of $P$ and different stationary distributions. In this paper, we consider the simple random walk that has a stationary distribution of $\pi(v) = k_v/(2|E|)$ for all $v \in V$.

\begin{definition} (\textbf{Simple Random Walk}).
Given a current node $v$, a simple random walk chooses uniformly at random a neighboring node $u \in N(v)$ and transit to $u$ in the next step - i.e.,
\begin{equation}
    P_{vu}=\left\{
    \begin{array}{ll}
        1/k_v\,\,\,&\text{if $u \in N(v)$,}\\
        0\,\,\,&\text{otherwise.}
    \end{array}
    \right.
\end{equation}
\end{definition}

One can see that each step of a simple random walk requires exactly one query (i.e., $q(v)$ to identify the neighborhood of $v$ and select the next stop $u$). Thus, the performance of sampling - i.e., the tradeoff between bias and query cost - is determined by how fast the random walk converges to the stationary distribution. Formally, we measure the convergence speed as the {\em mixing time} defined as follows.

\begin{definition} (\textbf{Mixing Time}) \label{def:rpd}
Given $G:(V, E)$, after $t$ steps of simple random walk, the {\em relative point-wise distance} between the current sampling distribution and the stationary distribution is
    \begin{align}
    \bigtriangleup(t)=\mathop{\max}_{u, v\in V, v\in N(u)}\left\{\frac{|P_{uv}^{t}-\pi(v)|}{\pi(v)}\right\} 
    \end{align}
where $P_{uv}^{t}$ is the element of $P^t$ with indices $u$ and $v$. The {\em mixing time} of the random walk is the minimum value of $t$ such that $\bigtriangleup(t) \leq \epsilon$ where $\epsilon$ is a pre-determined threshold on relative point-wise distance.
\end{definition}
One can see from the definition that the relative point-wise distance $\bigtriangleup(t)$ measures the bias of the random walk after $t$ steps. Mixing time, on the other hand, captures the query cost required to reduce the bias below a pre-determined threshold $\epsilon$. In the following subsection, we describe a key characteristics of the graph that determines the mixing time - the conductance of the graph.




\subsection{Conductance: An Efficiency Indicator}
Intuitively, the conductance $\Phi$, which indicates how fast the simple random walk converges to its stationary distribution, measures how ``well-knit'' a graph is. Specifically, the conductance is determined by a {\em cut} of the graph $G$ - i.e., a partition of $V$ into two disjoint subsets $S$ and $\bar{S}$ - which minimizes the ratio between the probability for the random walk to move from one partition to the other and the probability for the random walk to stay in the same partition. Formally, we have the following definition.

\begin{definition} \label{def:con} (\textbf{Conductance}). The conductance\footnote{Rigidly, the conductance is determined by both the graph topology and the transition matrix of the random walk. Here we tailor the definition to the simple random walk considered in this paper.} of a graph $G:(V, E)$ is
    \begin{equation}
        \Phi(G) = \mathop{\min}_{S\subseteq V}\frac{|\{e_{uv} | u \in S, v \in \bar{S}\}|}{\min\left\{|\{e_{uv} | u \in S, v \in V\}|,|\{e_{uv} | u \in \bar{S}, v \in V\}|\right\}}.
     \nonumber
    \end{equation}
\end{definition}
The relationship between the graph conductance and the mixing time of a simple random walk is illustrated by the following inequality \cite{Alon1986}:
\begin{align}
    (1-2\Phi(G))^t \leq \bigtriangleup(t) \leq \frac{2|E|}{\min_{v\in V}k_v} \left(1-\frac{\Phi(G)^2}{2}\right)^t.
    \label{conductance_inequation}
\end{align}
One can see that the graph conductance $\Phi(G)$ ranges between 0 and 1 - and the larger $\Phi(G)$ is, the smaller the mixing time will be (for a fixed threshold $\epsilon$). Also note from (\ref{conductance_inequation}) the log scale relationship between $\Phi(G)$ and the mixing time. This indicates a small change on $\Phi(G)$ may lead to a significant change of the mixing time. Let
\begin{align}
\frac{2|E|}{\min_{v\in V}k_v}& \left(1-\frac{\Phi(G)^2}{2}\right)^t \leq \epsilon \\
\Rightarrow t &\geq \frac{1}{\log (1-\Phi(G)^2)}\log\left(\frac{\epsilon}{\frac{2|E|}{\min_{v\in V}k_v }}\right) \\
\Rightarrow t &\geq -\frac{1}{\log (1-\Phi(G)^2)}\log(c/\epsilon)
\end{align}
Here $c = \frac{2|E|}{\min_{v\in V}k_v }$. For example, increasing conductance from 0.010 to 0.012 will change the mixing time from $46050.5\cdot \log(c/\epsilon)$ to $31979.1\cdot \log(c/\epsilon)$.


\vspace{1mm}
\noindent \framebox[\columnwidth]{\parbox{0.9\columnwidth}{{\bf Running Example}:
The conductance of the barbell graph in the running example is $\Phi(G) = 1/(\binom{11}{2}+1) = 0.018$. The corresponding (and unique) $S$ and $\bar{S}$ are shown in Fig~\ref{doublelayer}. Correspondingly, the mixing time to reach a relative point-wise distance of $\bigtriangleup(t) \leq \epsilon$ is bounded from above by $14212.3 \cdot \log (22.2/\epsilon)$. We shall show throughout the paper how our on-the-fly topology modification techniques can significantly increase conductance and reduce the mixing time for this running example.
}}
\vspace{1mm}

\subsection{Key for Conductance: Cross-Cutting Edges}
\label{sec:obj_mto}

A key observation from Definition~\ref{def:con} is that the graph conductance critically depends on the number of edges which ``cross-cut'' $S$ and $\bar{S}$ - i.e., $|\{e_{uv} | u \in S, v \in \bar{S}\}|$. The more such cross-cutting edges there are, the higher the graph conductance is likely to be. On the other hand, since a non-cross-cutting edge is only counted in the denominator, the more non-cross-cutting edges there are in the graph, the lower the conductance is likely to be. Formally, we define cross-cutting edges as follows.

\begin{definition}(\textbf{Cross-cutting edges}). 
    For a given graph $G(V,E)$, an edge $e_{uv}$ is a cross-cutting edge if and only if there exists $S \subseteq V$ such that $u \in S$, $v \in \bar{S}$ where $\bar{S} = V \backslash S$, and
    \begin{align}
        \varphi(S) = \frac{|\{e_{uv} | u \in S, v \in \bar{S}\}|}{\min\left\{|\{e_{uv} | u \in S, v \in V\}|,|\{e_{uv} | u \in \bar{S}, v \in V\}|\right\}} \nonumber
    \end{align}
takes the minimum value among all possible $S \subseteq V$. 
\end{definition}
We note that in large graphs such as online social networks, it is reasonable to assume that the number of cross-cutting edges is relatively small when compared to total number of edges in $S$ or $\bar{S}$.

One can see that our objective of on-the-fly topology modification is then to increase the number of cross-cutting edges and decrease the number of non-cross-cutting edges as much as possible. We describe our main ideas for doing so in the next section.

\vspace{1mm}
\noindent \framebox[\columnwidth]{\parbox{0.9\columnwidth}{{\bf Running Example}:
For the barbell graph, adding any edge between the two halves of the graph produces a new cross-cutting edge, and increases the graph conductance from $\Phi(G) = 0.018$ to $0.035$ - i.e., the mixing-time will be reduced to 3758.1/14212.3 = 0.264 - a significant reduction of 75\%.
}}
\vspace{1mm}

\section{Main Ideas of On-The-Fly Topology Modification}


\subsection{Technical Challenges: Negative Results}
\label{sec:identification_challenge}

One can see from Section~\ref{sec:obj_mto} that the key for increasing the conductance of a social network (and thereby reducing the query cost of sampling) through topology modification is to determine whether an edge is a cross-cutting edge or not. Unfortunately, the {\em deterministic} identification of a cross-cutting edge is a hard problem (in the worst case) even if the entire graph topology is given as prior knowledge, as shown in the following theorem.

\begin{theorem}
The problem of determining whether an edge is cross-cutting or not is NP-hard.
\end{theorem}
\begin{proof}
Consider the case of equal transition probability for each edge. The problem of finding all cross-cutting edges is equivalent with finding the optimum cut of the graph according to the Cheeger constant - a problem proved to be NP-hard \cite{Chung2007b}.
\end{proof}


%

Given the worst-case hardness result, we now consider the best-case scenario - i.e., is there any graph topology (which is not the worst-case input, of course) for which it is possible to efficiently identify cross-cutting edges? It is easy to see that, if the entire graph topology is given, then there certainly exist such graphs - with the original graph in Fig~\ref{doublelayer} being an example - for which the cross-cutting edge(s) can be straightforwardly identified. Nonetheless, our interest lies on making such identifications based solely upon local neighborhood knowledge - because of the aforementioned restrictions of online social-network interfaces. The following theorem, unfortunately, shows that it is impossible for one to deterministically {\em confirm} the cross-cutting nature of an edge unless the entire graph topology has been crawled.

\begin{theorem} \label{thm:nr2}
Given the local neighborhood topology of vertices accessed by a third-party sampler, $\{v_1,$ $\ldots, v_k\} \subset V$ in $G(V, E)$ where $k < |V|$, for any given edge $e: (v_i, v_j)$, there must exist a graph $G^\prime(V^\prime, E^\prime)$ such that: (1) $e: (v_i, v_j)$ is not a cross-cutting edge for $G^\prime$, and (2) $G$ and $G^\prime$ are indistinguishable from the view of the sampler - i.e., there exists $\{v^\prime_1,$ $\ldots, v^\prime_k\} \subset V^\prime$ which have the exactly same local neighborhood as $\{v_1,$ $\ldots, v_k\}$.
\end{theorem}
\begin{proof}
The construction of $G^\prime$ can be stated as follows: First, insert $n$ extra vertices $v^0_1, \ldots, v^0_n$ and $e$ extra edges into the graph, such that $\forall e:(v_i, v_j) \in E$, there is $e^0:(v^0_i, v^0_j)$ in the new graph. Note that at this moment, there is no edge between any $v_i$ and $v^0_j$. Then, in the second step, identify from $G$ a vertex $w$ which has not been accessed by the sampler - i.e., $w \not \subseteq \{v_1, \ldots, v_k\}$ - and insert into the graph an edge $e:(w, w^0)$. One can see that the only cross-cutting edge in the output graph $G^\prime$ is $(w, w^0)$ - i.e., $e: (v_i, v_j)$ cannot be a cross-cutting edge for $G^\prime$. An intuitive demonstration of the proof is shown in Fig~\ref{fig:adding_edge}.
\end{proof}

It is important to note from the theorem, however, that it still leaves two possible ways for one to increase the conductance of a social network based on only the local neighborhood knowledge: (1) While the theorem indicates that it is impossible to deterministically {\em confirm} the cross-cutting nature of an edge, it may still be possible to deterministically {\em disprove} an edge from being cross-cutting - i.e., we may prove that an edge is definitely non-cross-cutting based on just local neighborhood knowledge, and therefore {\em remove} it to increase the conductance deterministically. (2) It is still possible to conditionally or probabilistically evaluate the likelihood of an edge being cross-cutting - e.g., we may determine that an edge absent from the original graph is more likely to be a cross-cutting edge (if added) than an existing edge, and thereby {\em replace} the existing edge with the new one to increase the conductance in a probabilistic fashion. We consider the removal and replacement strategies, respectively, in the next two subsections.

\begin{figure}
    \centering
    \includegraphics[width=3.4in]{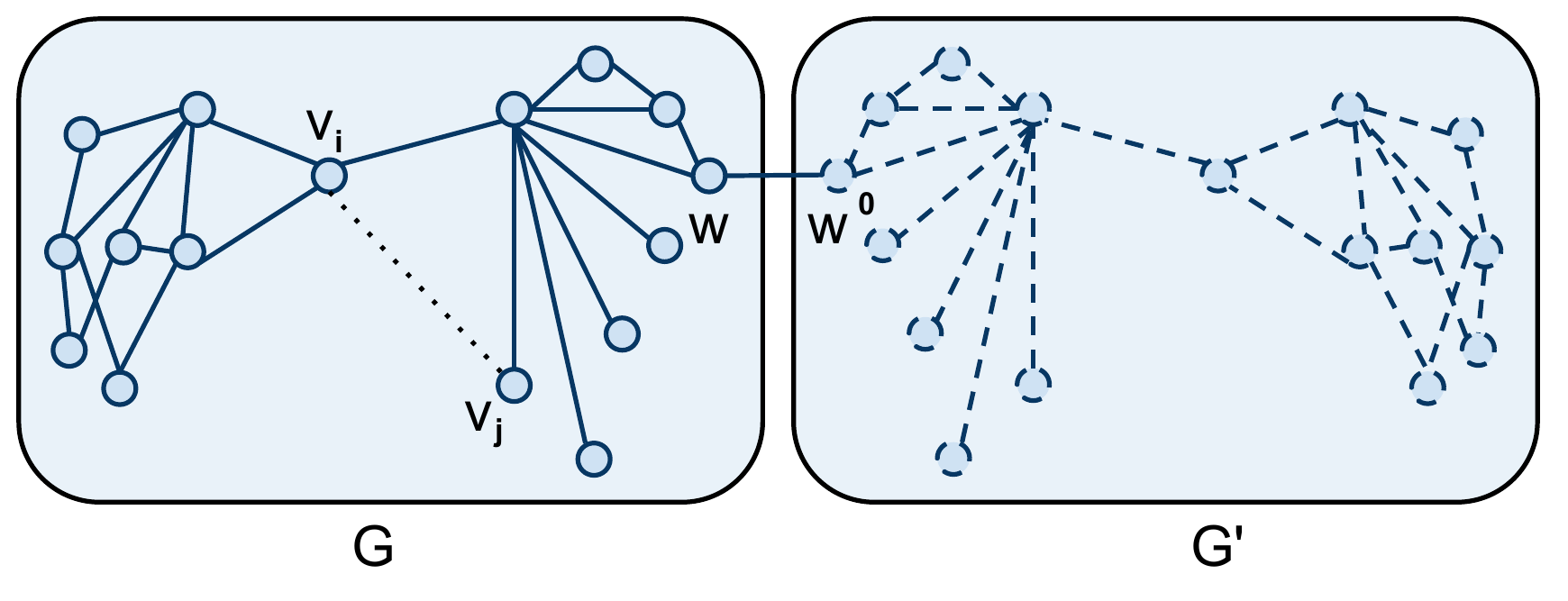}
    \caption{By cloning graph $G$, we can always construct graph $G^\prime$ such that simply adding an edge $e: (v_i,v_j)$ may decrease the conductance.}
    \label{fig:adding_edge}
\end{figure}


\subsection{Deterministic Identification of Non-cross-cutting Edges}
To illustrate the main idea of our deterministic identification of non-cross-cutting edges (for removals), we start with an example in Fig~\ref{fig:delete} to show why we can determine, based solely upon the local neighborhoods of $u$ and $v$ as shown in the graph, that $e: (u, v)$ (henceforth denoted by $e_{uv}$) in the Fig is not a cross-cutting edge. The intuition behind this is fairly simple: When $u$ and $v$ share a large number of common neighbors (e.g., 5 in Fig~\ref{fig:delete}) but have relatively few other edges (e.g., 1 each in Fig~\ref{fig:delete}), it is highly unlikely for the partition to cut through $e_{uv}$ rather than the other edges of $u$ and $v$ - e.g., $(u, u_0)$ in Fig~\ref{fig:delete} - if it cuts through any edges associated with $u$ and $v$ at all.

\begin{figure}
    \centering
    \includegraphics[width=3.2in]{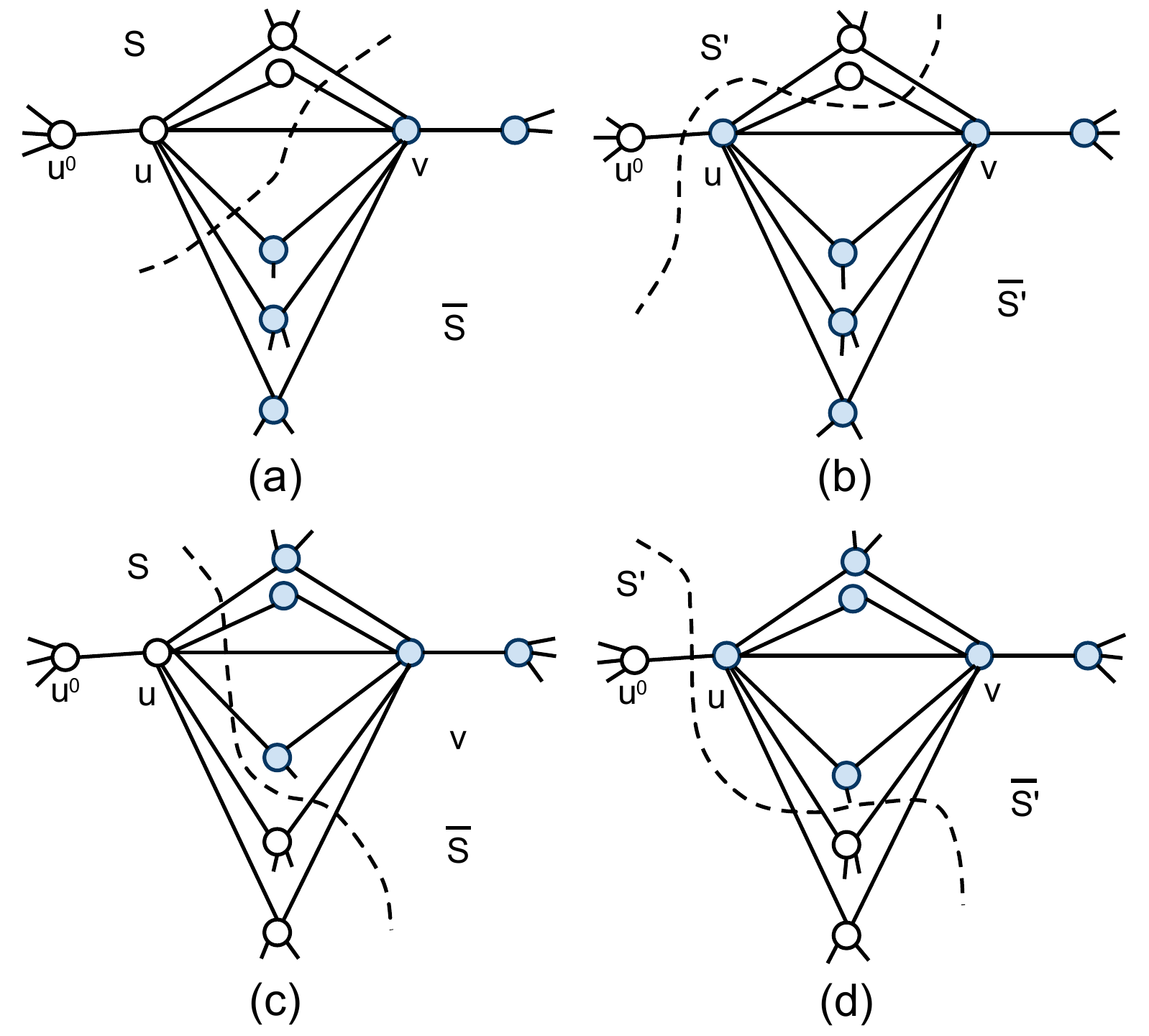}
    \caption{A figure shows that the edge $e_{uv}$ cannot be the cross-cutting edge in theorem \ref{delete_propo1}. Locally, (a) and (c) have 6 cross-cutting edges, while (b) and (d) only have 5 of them.}
    \label{fig:delete}
\end{figure}

The rigid (dis-)proof can be constructed with contradiction. Suppose $e_{uv}$ is a cross-cutting edge between two partitions of the graph, $S$ and $\bar{S}$. One can see that since $u$ and $v$ belong to different partitions, there must be at least 6 cross-cutting edges in the subgraph (Fig~\ref{fig:delete} (a) depicts an example). We now show in the following discussion that this is actually impossible because one can always construct another partition $S^\prime$ and $\bar{S}^\prime$ (by ``dragging'' $u$ and $v$ into the same part) and reduce the number of cross-cutting edges to at most 5. Note that this contradicts the definition of $S$ and $\bar{S}$ being a configuration which minimizes the number of cross-cutting edges. Thus, $e_{uv}$ cannot be a cross-cutting edge.

To understand how the construction of $S^\prime$ and $\bar{S}^\prime$ works, consider Fig \ref{fig:delete} (b) as an example. For the partition illustrated in Fig~\ref{fig:delete} (a), we can ``drag'' $u$ into $\bar{S}$ to form the new configuration, such that the number of cross-cutting edges associated with $u$ and $v$ is now at most 5, as shown in Fig \ref{fig:delete} (b). Note that the other edges not shown in the subgraph (no matter cross-cutting or not) are not affected by the re-configuration, because all vertices associated with $u$ are already known in the local neighborhood of $u$ (shown in Fig~\ref{fig:delete}). 

More generally, for the other possible settings of $S$ and $\bar{S}$ (such as Fig~\ref{fig:delete}(c)), one can construct the re-configuration in analogy with the following general principle: First, find the ``more popular'' partition (i.e., either $S$ or $\bar{S}$) among the 5 common neighbors of $u$ and $v$ (e.g., $\bar{S}$ in Fig~\ref{fig:delete} (a) or Fig~\ref{fig:delete} (c)). Then, drag one of $u$ and $v$ to ensure that both of them are in this more popular partition under the new configuration. One can see that, since at most 2 common neighbors of $u$ and $v$ are in the less popular partition, the number of cross-cutting edges under the new configuration is at most $2*2 + 1$, where $2*2$ is the number of cross-cutting edges associated with the 2 common neighbors in the less popular partition (at most 2 for each), and 1 is the number of cross-cutting edge associated with the other (non-common) neighbor of the node being dragged (i.e., $u^0$ in Fig~\ref{fig:delete} (a)).


The following theorem depicts the general case for which we can remove an edge on-the-fly to increase the graph conductance. Recall that $N(u)$ and $k_{u}$ represent the set of neighbors and the degree of a node $u$, respectively. 

\begin{theorem} {\em [Edge Removal Criteria]:} Given $G(V, E)$, $\forall u, v \in V$, if $e_{uv}\in E$ and
    \begin{equation}\left\lceil\frac{|N(u)\cap N(v)|}{2}\right\rceil + 1 > \frac{1}{2}\max\{k_u,k_v\} ,
    \label{delete_equ1}    
    \end{equation}
    then $e_{uv}$ is not a cross-cutting edge.
    \label{delete_propo1}
\end{theorem}
\begin{proof} Let $n=|N(u)\cap N(v)|$, without losing generality, assuming $u\in S, v\in \bar{S}$, then there must be n cross-cutting edges in these n disjoint paths of length 2 between $u$ and $v$.
    We denote $n_u, n_v$ as the number of cross-cutting edges in these n paths connected with u and v, so $n_u+n_v=n$. One can see that if we try to ``drag'' $u$ from $u\in S$ to $u\in \bar{S}$, all the edges connected with $u$ would be modified, e.g. flip the edges linked to $u$, which means the old cross-cutting edges will be the new non-cross-cutting edges, and vice versa.
    As the assumption from inequality (\ref{delete_equ1}): $\left\lceil \frac{n}{2} \right\rceil +1 >  \frac{1}{2}\max\{k_u,k_v\}$, so either $n_u + 1  >  \frac{1}{2}k_u$ or $n_v + 1 >  \frac{1}{2}k_v$ holds. Without losing generality, assuming for vertex $u$ the inequality holds, we change $u$ from set $S$ to $\bar{S}$, so the number of cross-cutting edges must be strictly decreasing. Since we have assumed that the number of edges in $S$ or $\bar{S}$ is much greater than the number of cross-cutting edges, so $\Phi (G)$ must decrease according to the decrease of the number of cutting-edges, which leads to the contradiction of $e_{uv}$ is a cross-cutting edge.
\end{proof}

Due to space limitations, please refer to the technical report  \cite{technical_report} for the proofs of all theorems in the rest of the paper. Intuitively, theorem \ref{delete_propo1} gives us a clue that if two nodes have enough common neighbors, then we can {\em deterministically} say that the edge between them is non-cross-cutting. Moreover, (\ref{delete_equ1}) is tight - i.e., if it does not hold, then we can always construct a counter example where $e_{uv}$ is cross-cutting - as shown in the following theorem.


\newtheorem{corollary}{Corollary}
\begin{corollary} For all $N(u), N(v), k_u, k_v$ which satisfy 
\begin{equation}
\left\lceil \frac{|N(u)\cap N(v)|}{2} \right\rceil +1 \leq \frac{1}{2}\max\{k_u,k_v\},
\label{delete_equ2}
\end{equation} 
there always exists a graph $G(V, E)$ in which $e_{uv}$ is cross-cutting. 
\label{delete_propo2}
\end{corollary}

\vspace{1mm}
\noindent \framebox[\columnwidth]{\parbox{0.9\columnwidth}{{\bf Running Example}: With our on-the-fly edge removals, any random walk is essentially following an overlay topology $G^*$ which can be constructed by applying Theorem \ref{delete_propo1} to every edge in the original graph $G$. For the bar-bell running example, the solid lines in Fig~\ref{doublelayer} depicts $G^*$. The conductance is now $\Phi(G^*) = 0.053$. Compared with the original conductance of 0.018, the corresponding lower bound on mixing time is reduced to 1638.3/14212.3 = 0.115 of the original value - a reduction of $89\%$.}}
\vspace{1mm}

\subsection{Conditional Identification of Cross-cutting Edges}

\begin{figure}
    \centering
    \includegraphics[width=2.4in]{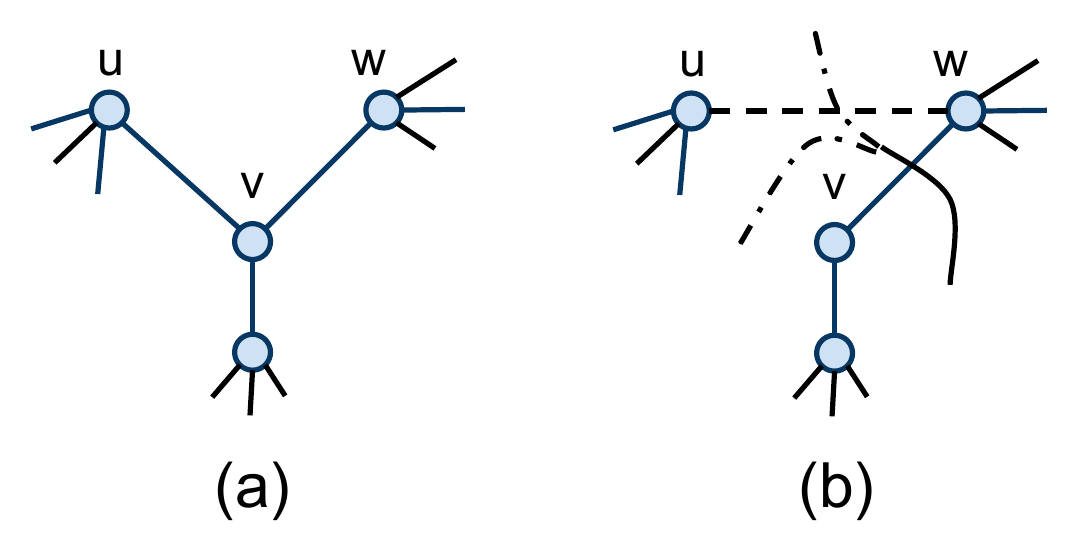}
    \caption{Replace the edge $e_{uv}$ with $e_{uw}$}
    \label{replace_fig}
\end{figure}

We now describe our second idea of conditionally identifying cross-cutting edges. We start with an example in Fig \ref{replace_fig} to show why we can {\em replace} an existing edge with a new one such that (1) the new edge is more likely to be crosscutting, and (2) the replacement is guaranteed to {\em not} decrease the conductance.

Specifically, consider the replacement of $e_{uv}$ by $e_{uw}$ given the neighborhoods of $u$ and $v$. A key observation here is that $e_{uv}$ and $e_{vw}$ cannot be both cross-cutting edges. The reason is that otherwise we could always ``drag'' $v$ into the same partition as $u$ and $w$ to reduce the number of cross-cutting edges by at least 1. Given this key observation, one can see that the replacement of $e_{uv}$ by $e_{uw}$ will only have two possible outcomes:
\begin{itemize}
\item if $e_{uv}$ is a cross-cutting edge, then $e_{uw}$ must also be a cross-cutting edge because, due to the observation, $e_{vw}$ cannot be a cross-cutting edge. Thus, the replacement leads to no change on the graph conductance.
\item if $e_{uv}$ is not a cross-cutting edge, then replacing it with $e_{uw}$ will either keep the same conductance, or increase the conductance if $e_{uw}$ is cross-cutting.
\end{itemize}
As such, the replacement operation never reduces the conductance, and might increase it when $e_{uw}$ is cross-cutting. More generally, we have the following theorem.

\begin{theorem}Given $G(V,E)$, $\forall v\in V$, if $k_v=3$, $u,w \in N(v)$, then replacing edge $e_{uv}$ with $e_{uw}$ will not decrease the conductance, while it also has positive possibility to increase the conductance. 
\label{replace_propo}
\end{theorem}
%

Next, we are going to prove that $k_v=3$ is actually the only case when replacement is guaranteed to not reduce the conductance, as shown by the following corollary.

\begin{corollary}\label{replace_tight} For $v\in V$, if $k_v \neq 3$, then there always exist a graph 
$G(V,E)$, $\forall u,w \in N(v)$, such that replacing $e_{uv}$ with $e_{uw}$ will decrease the conductance or have no effect.
\end{corollary}
%
%


\vspace{1mm}
\noindent \framebox[\columnwidth]{\parbox{0.9\columnwidth}{{\bf Running Example}: With Theorem \ref{replace_propo}, an example of the replacement operations one can perform over the bar-bell running example in Fig~\ref{doublelayer} is to replace $e_{ur}$ with $e_{rv}$, given that $u$ (after edge removals) has a degree of 3. Compared with the original conductance of $\Phi(G)$ = 0.018 and the post-removal conductance of $\Phi(G^{*})$ = 0.053, the conductance is now further increased to $\Phi(G^{**}) = 0.105$. The corresponding lower bound on mixing time is reduced to 416.6/1638.3 = 0.25 of the post-removal bound - a further reduction of 75\% - and 416.6/14212.3 = 0.029 of the original bound - an overall reduction of $97\%$.}}
\vspace{1mm}

\subsection{Extension}


If we know more about the user's neighbors, especially the common neighbors of the user and the random walk's next candidate, we will deterministically identify more non-cross-cutting edges. When the random walk reaches the nodes we have accessed before, we can use their degree information without issuing extra web requests since we could retrieve data from our local database. 

Fig \ref{delete_fig_propo3} (a) shows an example that with the extra degree knowledge of $u$ and $v$'s common neighbor $w$, $e_{uv}$ must be a non-cross-cutting edge. As $k_w=3$, if we assume $e_{uv}$ is a cross-cutting edges, then there must be 3 cross-cutting edges between $u$ and $v$. However, there exists another configuration Fig \ref{delete_fig_propo3} (b), which only has 2 cross-cutting edges. Thus, it contradicts the assumption that $e_{uv}$ is a cross-cutting edge. Noticed that if we do not know the degree of $w$, we could not deterministically identify $e_{uv}$ since theorem \ref{delete_propo1} does not apply here.
\begin{figure}
    \centering
    \includegraphics[width=2.7in]{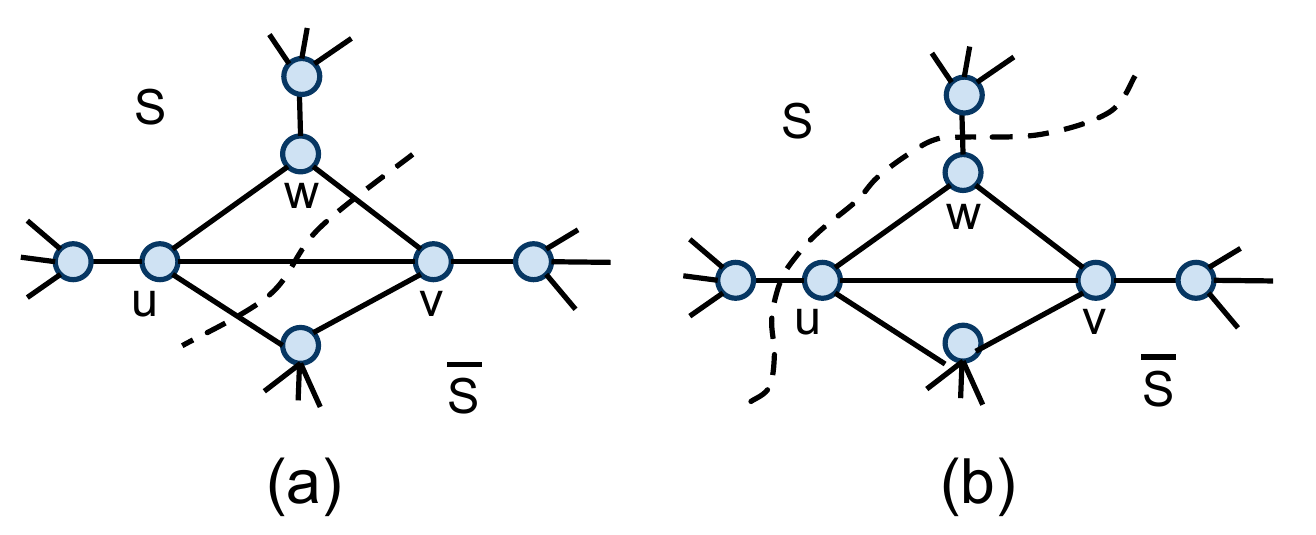}
    \caption{A demo shows that $e_{uv}$ cannot be a cross-cutting edge in theorem \ref{delete_propo3}.}
    \label{delete_fig_propo3}
\end{figure}

\begin{theorem}
Given $G(V, E)$, $\forall u, v \in V$, if $e_{uv}\in E$ and
\begin{equation}
\left\lceil\frac{|N(u)\cap N(v)|-N^*}{2}\right\rceil + 1 + \frac{1}{2}\sum_{w\in N^*}(4-k_w) > \frac{1}{2}\max\{k_u,k_v\},
\label{delete_equ3}
\end{equation}
we can assert that $e_{uv}$ is not a cross-cutting edge. Here we denote $N^* = \{w\in N(u)\cap N(v) | \,\,k_w\text{ is known }, 2\leq k_w\leq 3 \}$.
\label{delete_propo3}
\end{theorem}

Intuitively, the edge between two nodes which have many common neighbors has higher probability to be a non-cross-cutting edge. Also, it is easy for us to find these edges in online social networks. If a friend knows almost every other friends of a person, then this edge may be considered as non-cross-cutting edge according to theorem \ref{delete_propo1} and \ref{delete_propo3}.

\section{Algorithm MTO-Sampler}

\subsection{Algorithm implementation}
\begin{figure*}
    \centering
        \subfigure[Modified overlay graph $G^*$]{
            \label{fig:modified}
            \includegraphics[width=3in]{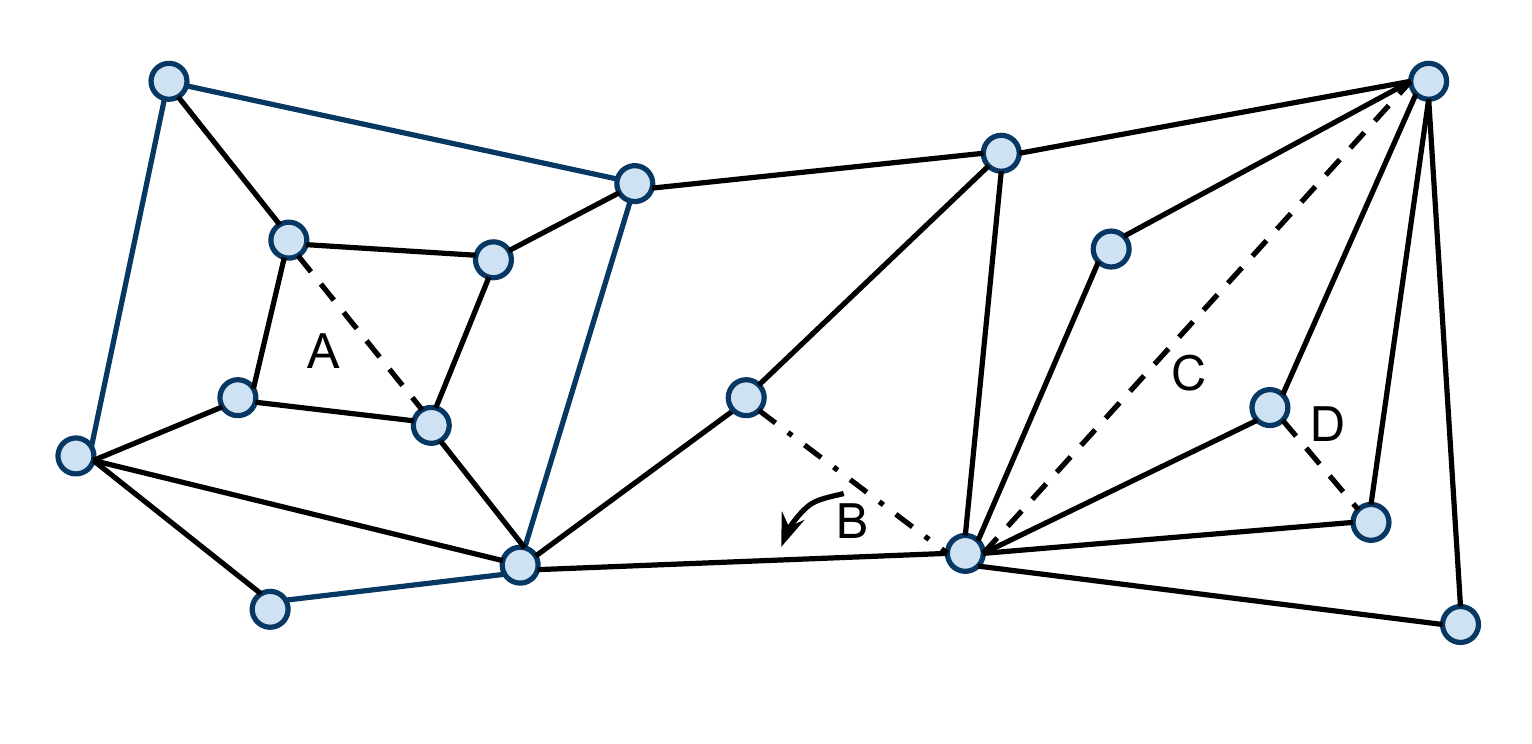}
        }
        \subfigure[Carry out the random walk by modify the topology on-the-fly. It is identical to the random walk in overlay graph $G^*$.]{
            \label{fig:mto}
            \includegraphics[width=3in]{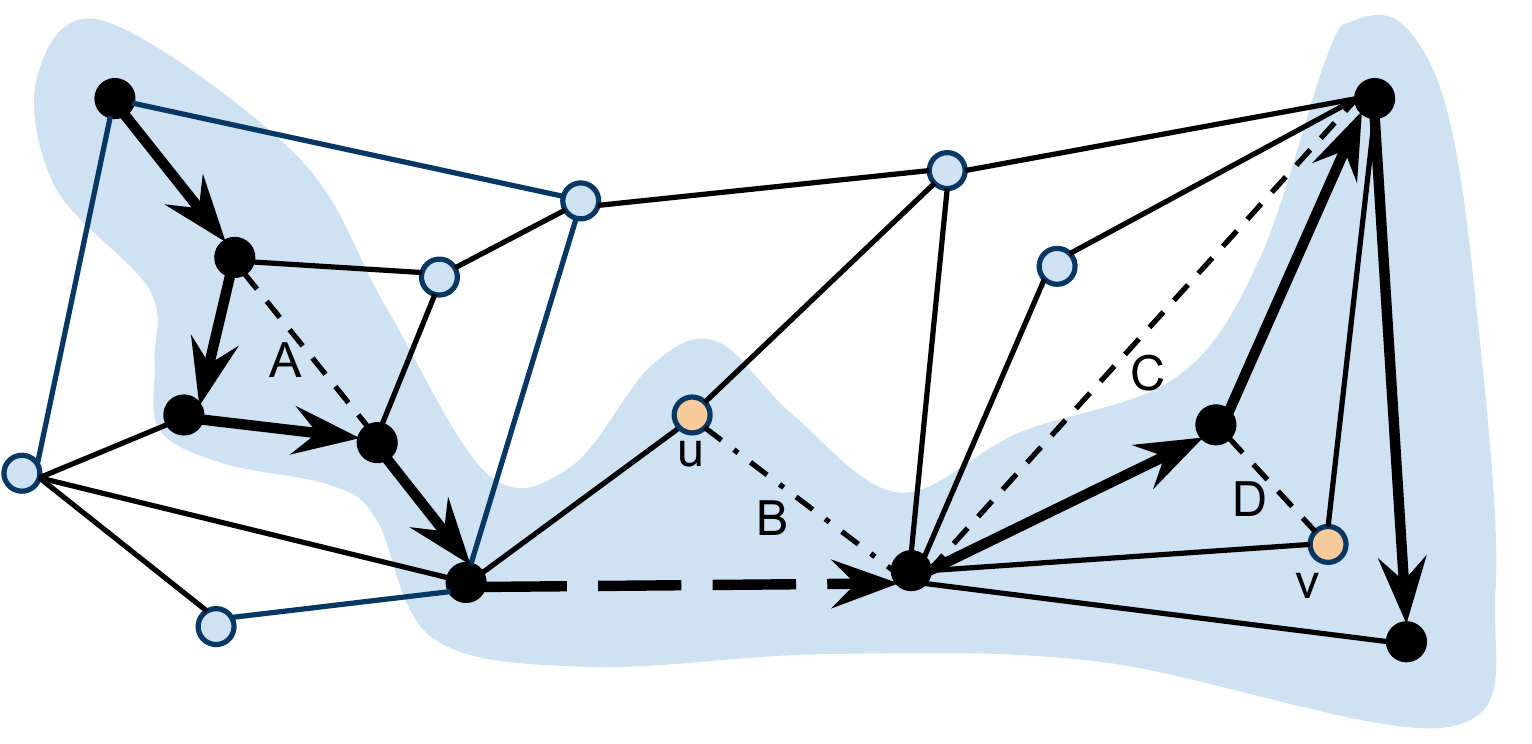}
        }
    \caption{A demo shows how the MTO-Sampler modifies the topology of the graph on-the-fly.}
    \label{demo}
\end{figure*}
\noindent\textbf{Algorithm description.}
To explain how the on-the-fly modification works, we demonstrate an example in Fig \ref{demo}. Fig \ref{fig:modified} is an overlay graph $G^*$ that has been modified according to former theorems, in which edges A, C and D are removed, and edge B is replaced. Fig \ref{fig:mto} shows one possible track of how our MTO-sampler change the simple random walk. For instance, when the random walk sees a node $u$, and $k_u=3$ (it satisfies the condition of replacement), then it may replace an edge as we described in theorem \ref{replace_propo}. The colored area contains all the nodes that the random walk visits. 

Algorithm 1 depicts the detailed procedure of MTO sampler, and the stopping rule (which indicates that the random walk should stop and output samples) can be any convergence monitor used in Markov Chain.

\begin{algorithm}
\label{al}
    \caption{MTO-Sampler for Simple Random Walk}
    \begin{algorithmic}
	        \FOR{$i = 1 \to sample\_size$}
	        \STATE Starting from vertex $u$
	        \WHILE{!(Stopping rule)}
	        \WHILE{$|N(u)|\geq 1$}
	        \STATE Uniformly pick a neighbor $v$, and issue a query
	        \IF{$e_{uv}$ is removable}
	        \STATE $N(u) \gets N(u) - \{v\}$
	        \STATE continue
	        \ELSIF{$k_v==3$}
	        \STATE /* One of $v$'s edge can be replaced*/
			\IF{choose to replace $e_{uv}$}
			\STATE $v \gets v'$
			\ELSE
			\STATE $N(u) \gets N(u) \cup \{v'\}$
			\STATE choose $u \gets v$ or $u \gets v'$ randomly
			\STATE break
	        \ENDIF
	        \ENDIF
	        \IF{rand$(0,1) < 1/2$ }
	        \STATE $u \gets v$
	        \STATE break
			\ELSE
			\STATE continue	        				        
			\ENDIF	        
	        \ENDWHILE
	        \ENDWHILE
	        \STATE Record sample $x_i \gets u$
	        \ENDFOR	        
	\end{algorithmic}
	\label{algorithm1}
\end{algorithm}



\noindent\textbf{Aggregate estimation and probability revision.}
After collecting samples, we use {\em Importance Sampling} to directly estimate the aggregate information through the samples from the random walk's stationary distribution $\tau$. 
    \begin{algorithm}
        \begin{algorithmic}
        \STATE \textbf{Importance Sampling: }
            \FOR{$i=1$ to Sample\_Size $N$}
            \STATE $x_i \gets$ sampling from $\tau$ 
			\STATE $w(x_i) \gets \frac{\hat{\pi}(x_i)}{\hat{\tau}(x_i)}$
			\STATE record $f(x_i)$ /*Aggregate Function $f(\cdot)$ */			
            \ENDFOR
		\STATE Output estimation $A(f(X)) = \frac{\frac{1}{N}\sum_{i=1}^Nf(x_i)w(x_i)}{\frac{1}{N}\sum_{i=1}^Nw(x_i)}$        
        \end{algorithmic}
    \end{algorithm}    

The key challenge for MTO-Sampler using {\em importance sampling} is to estimate the stationary distribution of MTO-Sampler random walk $\tau^*$. Since MTO-Sampler modifies the topology, $\tau^*$ may not equal to the stationary distribution $\tau$. Here we have
\begin{equation}\tau^*(u) = \frac{k_u^*}{2|E^*|}.\end{equation}
$k_u^*$ is unknown in overlay graph $G^*$, but we can draw simple random sample from $u$'s neighbors in $G^*$ to get an unbiased estimation of $k_u^*$.

\subsection{Theoretical Model Analysis}


In order to theoretically analysis the performance of MTO-Sampler, we introduce a well known graph generation model: {\em Latent space model}.

\noindent\textbf{Latent space model}. 
Latent space graph model \cite{Sarkar_theoreticaljustification} are connecting two nodes with the probability related to their distance in the latent space.
\begin{equation}P(i\sim j|d_{ij}) = \frac{1}{1+e^{\alpha(d_{ij}-r)}}, \end{equation}
here $d_{ij}$ is the distance between two nodes $i$ and $j$; $r$ controls the level of sociability of a node in this graph, and $\alpha$ is the sharpness of the function. 

We will show that in the following theorem if two nodes' distance is smaller than a threshold $d_0$, then it is likely to be an non-cross-cutting edge. Therefore, after finding the expected number of edges that can be removed we can calculate the increment of the conductance.

\begin{theorem}
Given a latent space graph model $G(V,E)$, assume $\alpha=+\infty$, then the expected number of edges we can removed
\begin{equation}
\mathbb{E}[R] \geq |E|\cdot\mathbb{P}\left(d<V(r)\left(1-\left(\frac{1}{3}\right)^{1/D}\right)\right),
\label{eq:E-R}
\end{equation}
here V(r) is the volume of a hypersphere with radius $r$ in $D$ dimensional latent space. The proof can be found at \cite{technical_report}. 
\end{theorem}

Simple simulations show that from 20000 points experiment, one can get the empirical distribution of point-wise distance. More specifically, If we let $r=0.7$, $a=4$ and $b=5$, $D=2$, then 
\begin{align}
\mathbb{E}[\Phi(G^*)] \geq 1.052\Phi(G)
\label{eq:ip_con}
\end{align}
We compared the experimental results together with this theoretical bound of latent space model in section \ref{sec:theoretical_exp}.

\section{Experiments}
\subsection{Experimental Setup}

\subsubsection{Hardware and Platform}
We conducted all experiments on a computer with Intel Core i3 2.27GHz CPU, 4GB RAM and 64bit Ubuntu operating system. All algorithms were implemented in Python 2.7. Our local, synthetic and online datasets are stored in the in-memory Redis database and the MongoDB database.


\begin{table}
\centering
\begin{tabular}{|r|r|r|r|r|r|r|r|r|r|r|r|r|r|}
\hline  Dataset & \#nodes & \#edges & 90\% diameter \\
\hline  Epinions\cite{Richardson2003} & 26588 & 100120 & 4.8\\
\hline  Slashdot A\cite{Leskovec2009} & 70068 & 428714 & 4.5\\
\hline  Slashdot B\cite{Leskovec2009} & 70999 & 436453 & 4.5\\
\hline
\end{tabular}
\caption{Local Datasets}
\label{tab:lsn}
\end{table}

\subsubsection{Datasets} \label{sec:epd}
We tested three types of datasets in the experiments: local real-world social networks, Google Plus online social network, and synthetic social networks - which we describe respectively as follows.

\vspace{2mm}
\noindent{\bf Local Datasets:} The local social networks - i.e., real-world social networks for which the entire topology is downloaded and stored locally in our server. For these datasets, we simulated the individual-user-query-only web interface strictly according to the definition in Section 1, and ran our algorithms over the simulated interface. The rationale behind using such local datasets is so as we have the ground truth (e.g., real aggregate query answers over the entire network) to compare against for evaluating the performance of our algorithms.

Table~\ref{tab:lsn} shows the list of local social networks we tested with (collected from \cite{stanford_dataset}). All three datasets are previously-captured topological snapshots of Epinions and Slashdot, two real-world online social networks. Since we focus on sampling undirected graphs in this paper, for a real-world directed graph (e.g., Epinions), we first convert it to an undirected one by only keeping edges that appear in both directions in the original directed graph. Note by following this conversion strategy, we guarantee that a random walk over the undirected graph can also be performed over the original directed graph, with an additional step of verifying the inverse edge (resp.~$v \to u$) before committing to an edge (resp.~$u \to v$) in the random walk. The number of edges and the 90\% effective diameter reported in Table~\ref{tab:lsn} represent values after conversion.

\vspace{2mm}
\noindent{\bf Google Plus Online Social Graph:} We also tested a second type of dataset: remote, online, social networks for which we have no access to the ground truth. In particular, we chose the Google Plus\footnote{https://plus.google.com/} network because its API\footnote{The source code of its Python wrapper can be found at https://github.com/pct/python-googleplusapi. After April 20, 2012, this social graph api will be fully retired.} is the most generous among what we tested in terms of the number of accesses allowed per IP address per day. Using random walk and MTO-Sampler random walk, we have accessed 240,276 users in Google Plus. We observed that the interface provided by Google Social Graph API strictly adheres to our model of an individual-user-query-only web interface, in that each API request returns the local neighborhood of one user. We also collected the data of users' self-description.

\vspace{2mm}
\noindent{\bf Synthetic Social Networks:} One can see that, for the real-world social network described above, we cannot change graph parameters such as size, connectivity, etc, and observe the corresponding performance change of our algorithms. To do so, we also tested synthetic social networks which were generated according to theoretical models. In particular, we tested the latent space model.

We note that, since the effectiveness of these theoretical models are still under research/debate, we tested these synthetic social networks for the sole purpose of observing the potential change of performance for social networks with different characteristics. The superiority of our algorithm over simple random walk, on the other hand, is tested by our experiments on the two types of real-world social networks.

\subsubsection{Algorithms Implementation and Evaluation} \label{sec:eps}

\noindent{\bf Algorithms:} We tested four algorithms, the simple random walk (i.e., baseline), Metropolis Hastings Random Walk (MHRW), Random Jump (RJ) and our MTO-Sampler,  and compared their performance over all of the above-described datasets. 

\vspace{2mm}
\noindent{\bf Input Parameters:} Both simple random walk and our MTO-sampler are parameter-less algorithms with one exception: They both need a {\em convergence indicator} to determine when the random walk has reached (or become sufficiently close to) the stationary distribution - so a sample can be retrieved from it. In the experiments, we used the Geweke indicator \cite{Geweke92evaluatingthe}, one of the most popularly used methods in the literature, which we briefly explain as follows.

Given a sequence of nodes retrieved by a random walk, the Geweke method determines whether the random walk reaches the stationary distribution after a burn-in of $k$ steps by first constructing two ``windows'' of nodes: Window A is formed by the first 10\% nodes retrieved by the random walk after the $k$-step burn-in period, and Window B formed by the last 50\%. One can see that, if the random walk indeed converges to the stationary distribution after burn-in, then the two windows should be statistically indistinguishable. This is exactly how the Geweke indicator tests convergence. In particular, consider any attribute $\theta$ that can be retrieved for each node in the network (a commonly used one is degree that applies to every graph). Let
\begin{align}
Z = \left|\frac{\bar{\theta}_A -\bar{\theta}_B }{\sqrt{\hat{S}_\theta^A + \hat{S}_\theta^B }}\right|,
\end{align}
where $\bar{\theta}_A$ and $\bar{\theta}_B$ are means of $\theta$ for all nodes in Windows $A$ and $B$, respectively, and $S_\theta^A$ and $S_\theta^B$ are their corresponding variances. One can see that $Z \to 0$ when the random walk converges to the stationary distribution. Thus, the Geweke indicator confirms convergence if $Z$ falls below a threshold. In the experiments, we set the threshold to be $Z \leq 0.1$ by default, while also performing tests with the threshold ranging from $0.01$ to $1$.

\vspace{2mm}
\noindent{\bf Performance Measures for Sampling:} As mentioned in Section~\ref{pf_measure}, a sampling technique for online social networks should be measured by {\em query cost} and {\em bias} - i.e., the distance between the (ideal) stationary distribution (i.e., $p(v) = deg(v)/\sum_v deg(v)$ for a simple random walk) and the actual probability distribution for each node to be sampled. To measure the query cost, we simply used the number of unique queries issued by the sampler. Bias, on the other hand, is more difficult to measure, as shown in the following discussions.

For a small graph, we measured bias by running the sampler for an extremely long amount of time (long enough so that each node is sampled multiple times). We then estimated the sampling distribution by counting the number of times each node is retrieved, and compared this distribution with the ideal one to derive the bias. In particular, we measured bias as the KL-divergence between the two distributions, specifically $D_\mathrm{KL}(P || P_\mathrm{sam}) + D_\mathrm{KL}(P_\mathrm{sam} || P)$, where $P$ and $P_\mathrm{sam}$ are the ideal distribution and the (measured) sampling distribution, respectively.

For a larger graph, one may need a prohibitively large number of queries to sample each node multiple times. To measure bias in this case, we use the collected samples to estimate aggregate query answers over all nodes in the graph, and then compare the estimation with the ground truth. One can see that, a sampler with a smaller bias tends to produce an estimation with lower relative error. Specifically, for the local social networks, we used the average degree as the aggregate query (as only topological information is available for these networks). For the Google Social Graph experiment, we tested various aggregate queries including the average degree and the average length of user self-description.

Finally, to verify the theoretical results derived in the paper, we also tested a theoretical measure: the mixing time of the graph. In particular, we continuously ran our MTO-Sampler until it hits each node at least once - so we could actually obtain the topology of the overlay graph (e.g., as in Fig~\ref{doublelayer}). Then, we computed the mixing time of the overlay graph (from the Second-Largest Eigenvalue Modulus (SLEM) of its adjacency matrix\footnote{Typical theoretical mixing time of Simple Random Walk can be defined as $\Theta(1/\log(1/\mu))$, where $\mu$ is SLEM of transition matrix $P$.}, see \cite{Boyd2005}). We would like to caution that, while we used it to verify our theoretical results of MTO-Sampler never decreasing the conductance of a graph, this theoretically computed measure does not replace the above-described bias vs.~query cost tests because it is often sensitive to a small number of ``badly-connected'' nodes (which may not cause significant bias for practical purposes).

%
%



\subsection{Performance Comparison Between Simple Random Walk and MTO-Sampler}

\begin{figure*}[!htb]
    \begin{center}
	    \subfigure[Slashdot A]{
	        \label{fig:slashdotA}
	        \includegraphics[width=0.31\linewidth]{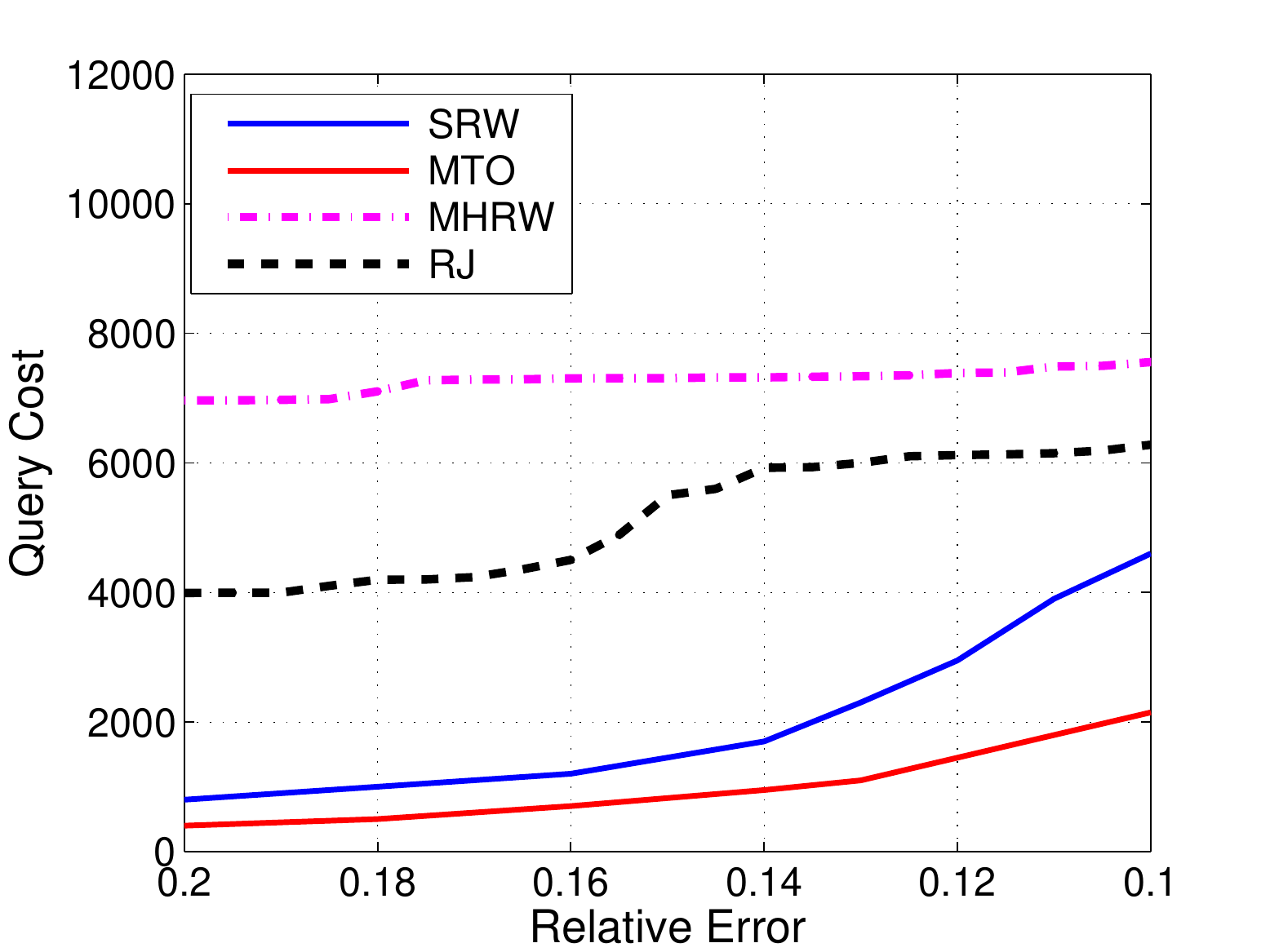}
        }
        \subfigure[Slashdot B]{
            \label{fig:slashdotB}
            \includegraphics[width=0.31\linewidth]{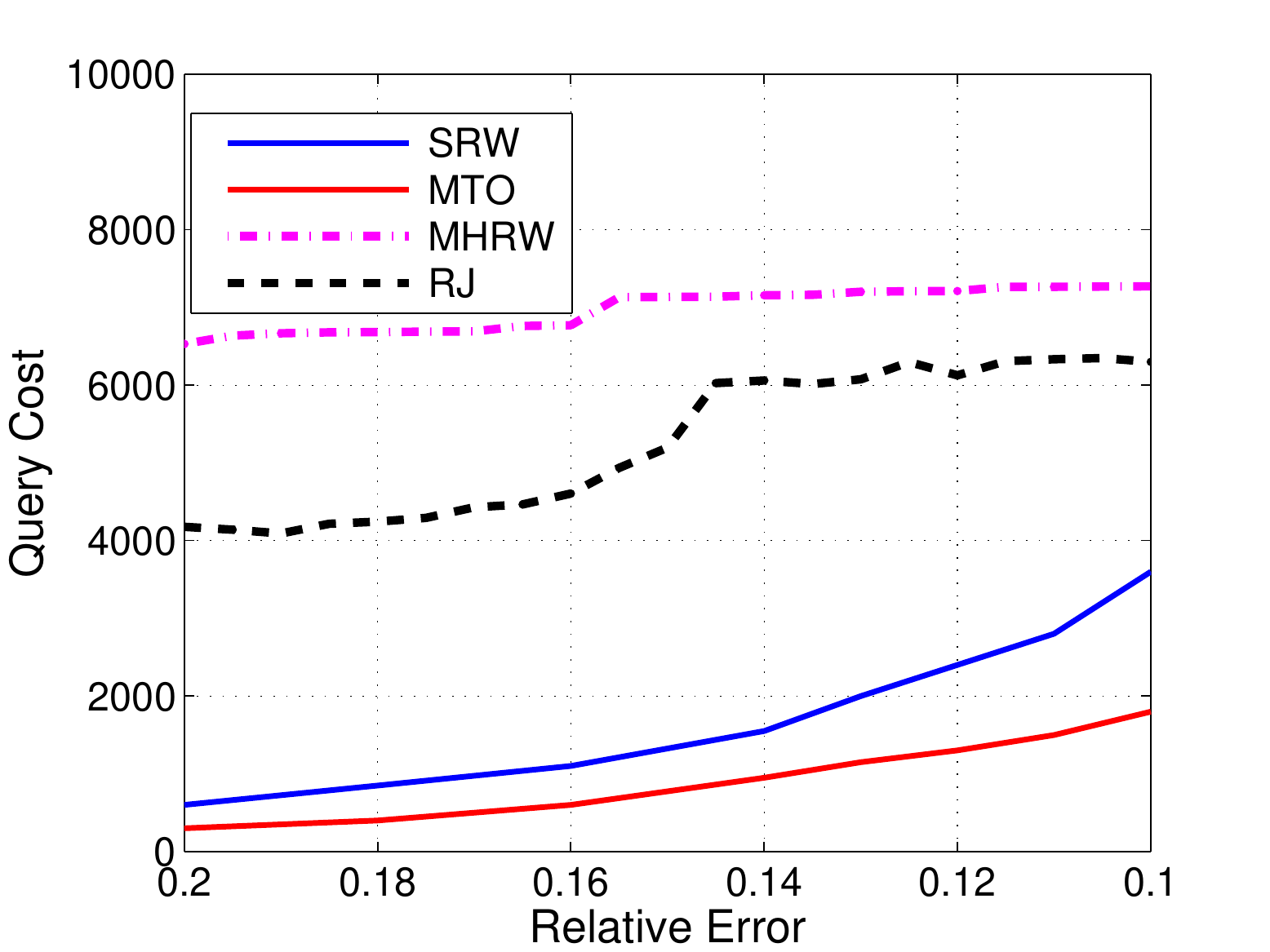}
        }
	    \subfigure[Epinions]{
	        \label{fig:epinion}
	        \includegraphics[width=0.325\linewidth]{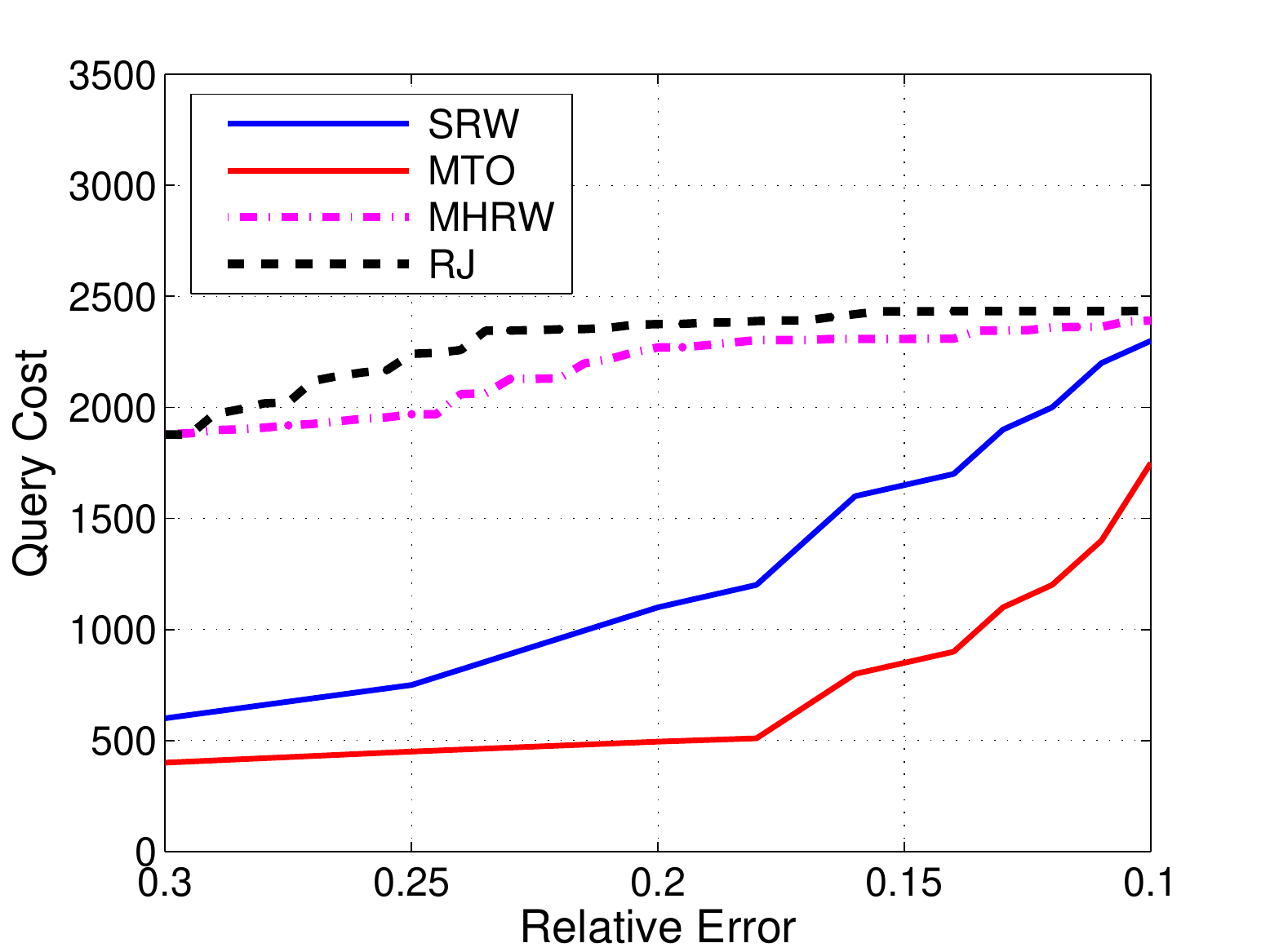}
        }
    \end{center}
    \caption{Bias vs. Query Cost tests for local datasets' average degree.}        
    \label{fig:guaranteed_qcost}
\end{figure*}

We started by comparing the performance of Simple Random Walk (SRW) and MTO-Sampler over real-world social networks using all three performance measures described above - KL-divergence, relative error vs.~query cost, and theoretical mixing time.

\vspace{2mm}
\noindent{\bf Local Datasets:} We started by testing the relative error vs.~query cost tradeoff of SRW, MTO, MHRW and RJ for estimating aggregate query answers. Since only topological information is available for local datasets, we used the average degree as the aggregate query. Fig~\ref{fig:guaranteed_qcost} depicts the performance comparison for the three real-world social networks. Here each point represents the average of 20 runs of each algorithm, and the query cost (i.e., y-axis) represents the maximum query cost for a random walk to generate an estimation with relative error above a given value (i.e., x-axis). For random jump in the experiments, we set the probability of jumping as 0.5. One can see that, for all three datasets, our MTO-Sampler achieves a significant reduction of query cost compared with the SRW sampler, MHRW sampler and Random Jump sampler.

\begin{figure*}
    \minipage{0.32\textwidth}
    \includegraphics[width=\linewidth]{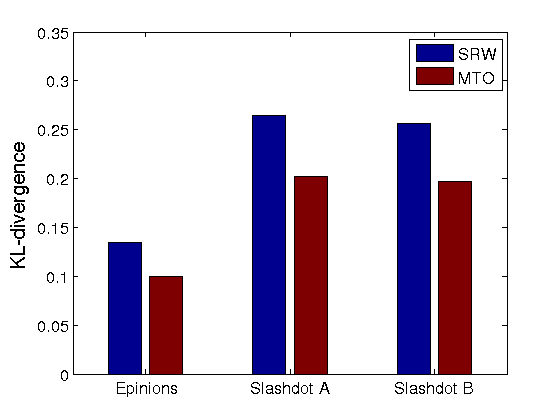}
    \caption{Comparison between SRW and MTO on query cost and the Kullback–Leibler divergence measure defined in Section~\ref{sec:eps} over all three datasets.}
    \label{fig:kl}
    \endminipage\hfill
    \minipage{0.32\textwidth}
    \includegraphics[width=\linewidth]{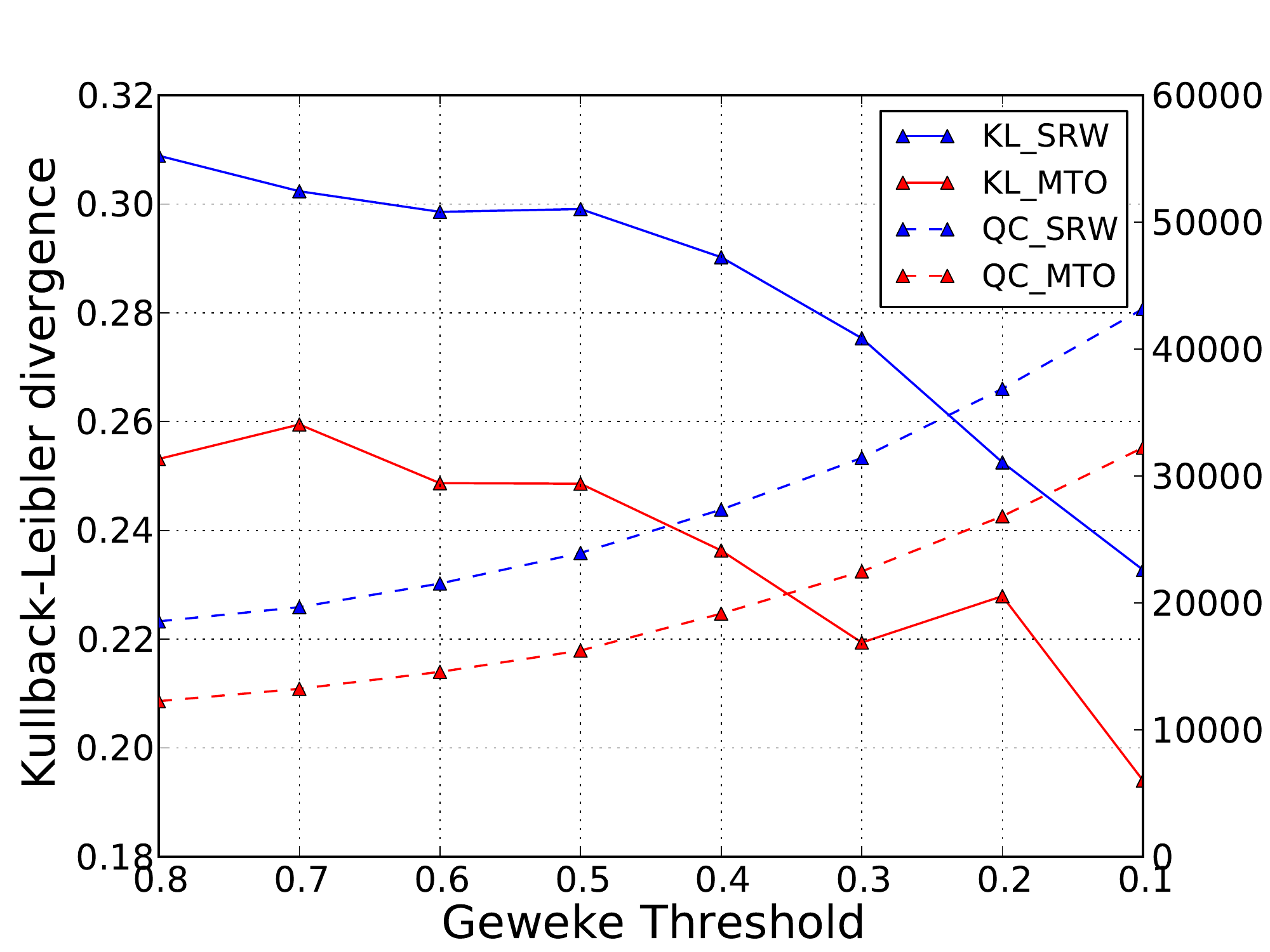}
    \caption{Varying Geweke Threshold to get different KL divergence on dataset Slashdot B. KL and QC stands for KL divergence and Query Cost respectively.}
    \label{fig:geweke_kl}
    \endminipage\hfill
    \minipage{0.32\textwidth}
    \includegraphics[width=\linewidth]{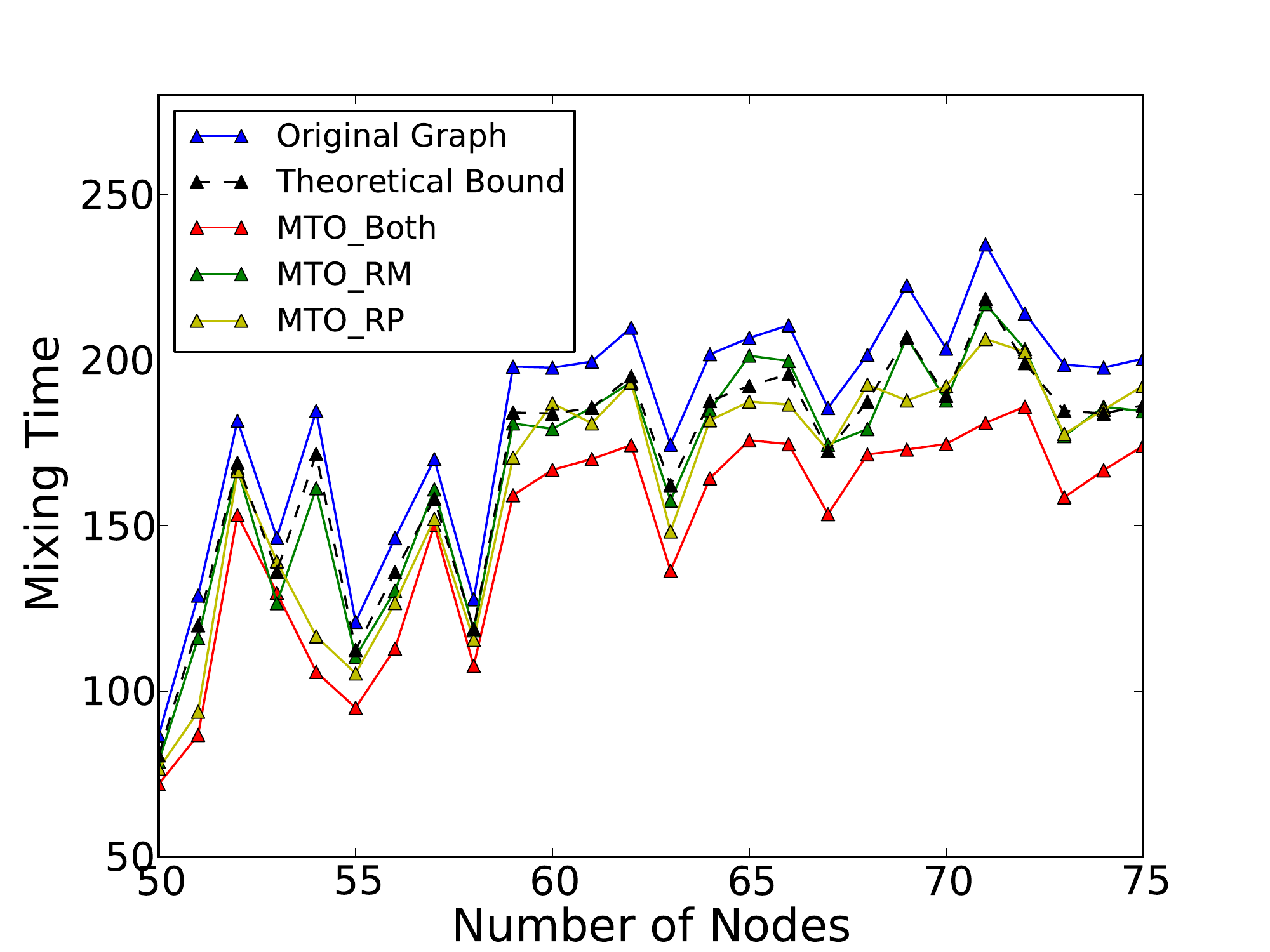}
    \caption{Comparison of theoretical mixing time on latent space graph model. MTO\_Both: Remove and replace edges. MTO\_RM: Only remove edges. MTO\_RP: Only replace edges.}
    \label{graph_eig}
    \endminipage
\end{figure*}

We also tested the KL-divergence measured by performing an extremely long execution of SRW and MTO in Fig~\ref{fig:kl} - with each producing 20000 samples - to estimate the sampling probability for each node. The Geweke threshold was set to be 0.1 for the test. One can see that our MTO-Sampler not only requires fewer queries for generating each sample (i.e., converges to the stationary distribution faster), but also produces less bias than the SRW sampler.

To further test the bias of samples generated by our MTO-Sampler, we also conducted the test while varying the Geweke threshold from 0.1 to 0.8 on the dataset Slashdot B. Fig~\ref{fig:geweke_kl} depicts the change of measured bias for SRW and MTO, respectively. One can see from the figure that our MTO-Sampler achieves smaller bias than SRW for all cases being tested. In addition, a smaller threshold leads to a smaller bias and larger query cost, as indicated by the definition of Geweke convergence monitor.


\vspace{2mm}
\noindent{\bf Google Plus online social network:} For Google Plus, we do not have the ground truth as the entire social network is too large (about 85.2 million users in Feb 2012\footnote{Estimated by Paul Allen's model, http://goo.gl/nZCzN}) to be crawled. Thus, we performed the tests in two steps. First, we continuously ran each sampler until their Geweke convergence monitor indicated that it had reached its stationary distribution. We then used the final estimation as the presumptive ground truth which we refer to as the {\em converged value}. In the second step, we used the converged value to compute the relative error vs.~query cost tradeoff as previously described.

\begin{figure*}
    \begin{center}
        \subfigure[Estimated Average Degree]{
            \label{fig:gp_degree}
            \includegraphics[width=0.31\textwidth]{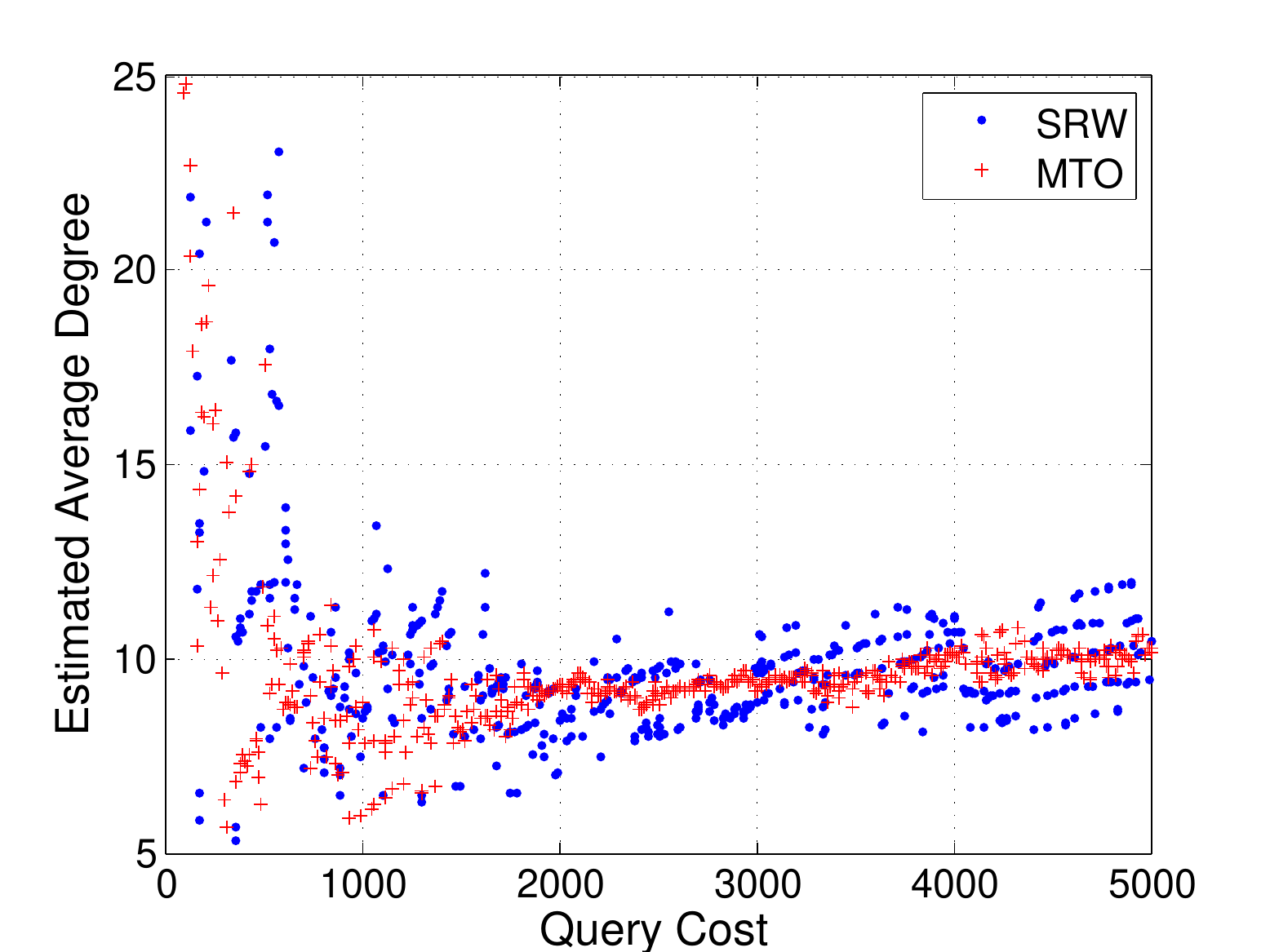}
        }
	    \subfigure[Average Degree]{
	        \label{fig:gp_error_deg}
	        \includegraphics[width=0.31\textwidth]{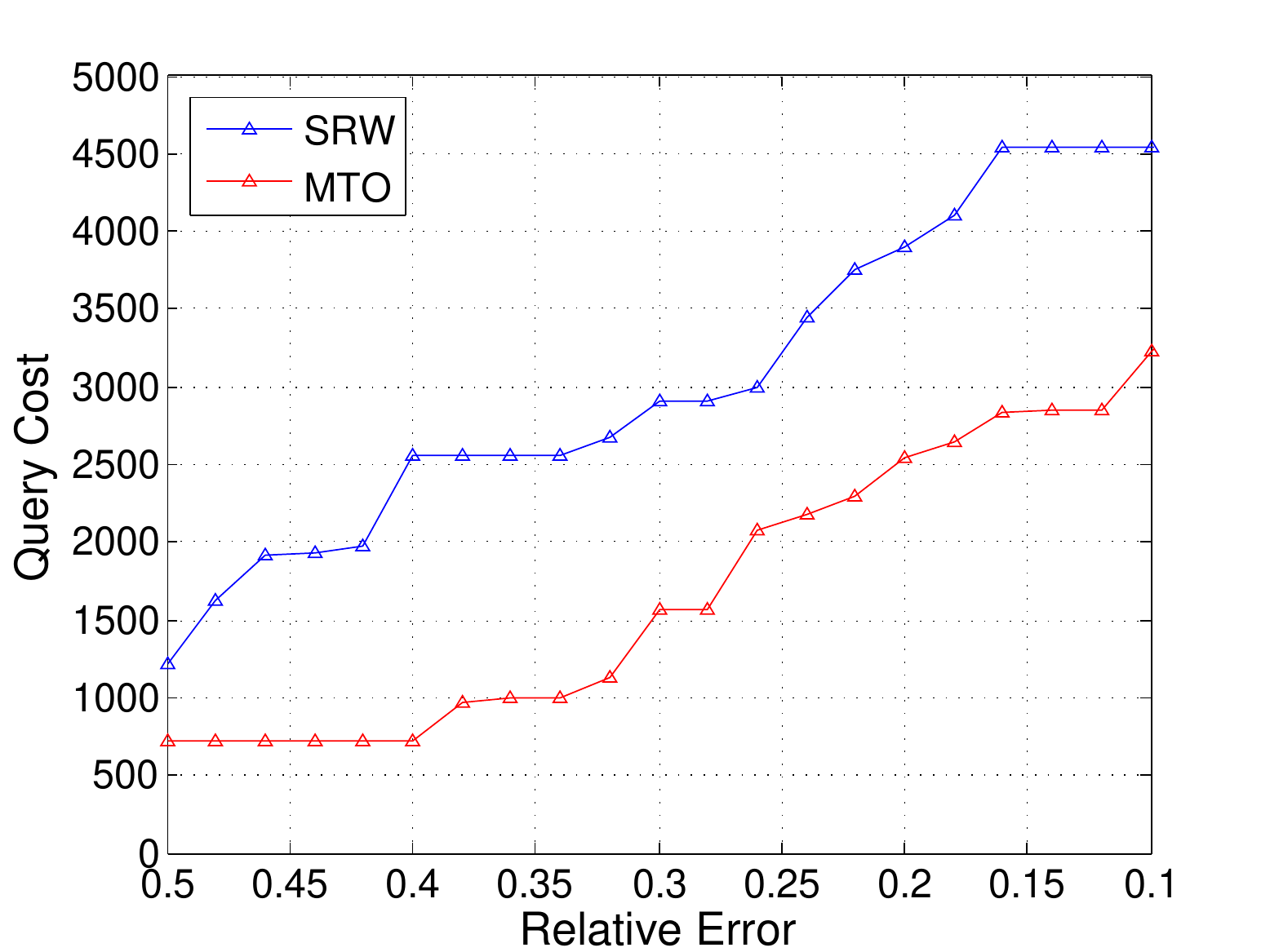}
        }
        \subfigure[Average Self-description Length]{
	        \label{fig:gp_introlen}
	        \includegraphics[width=0.31\textwidth]{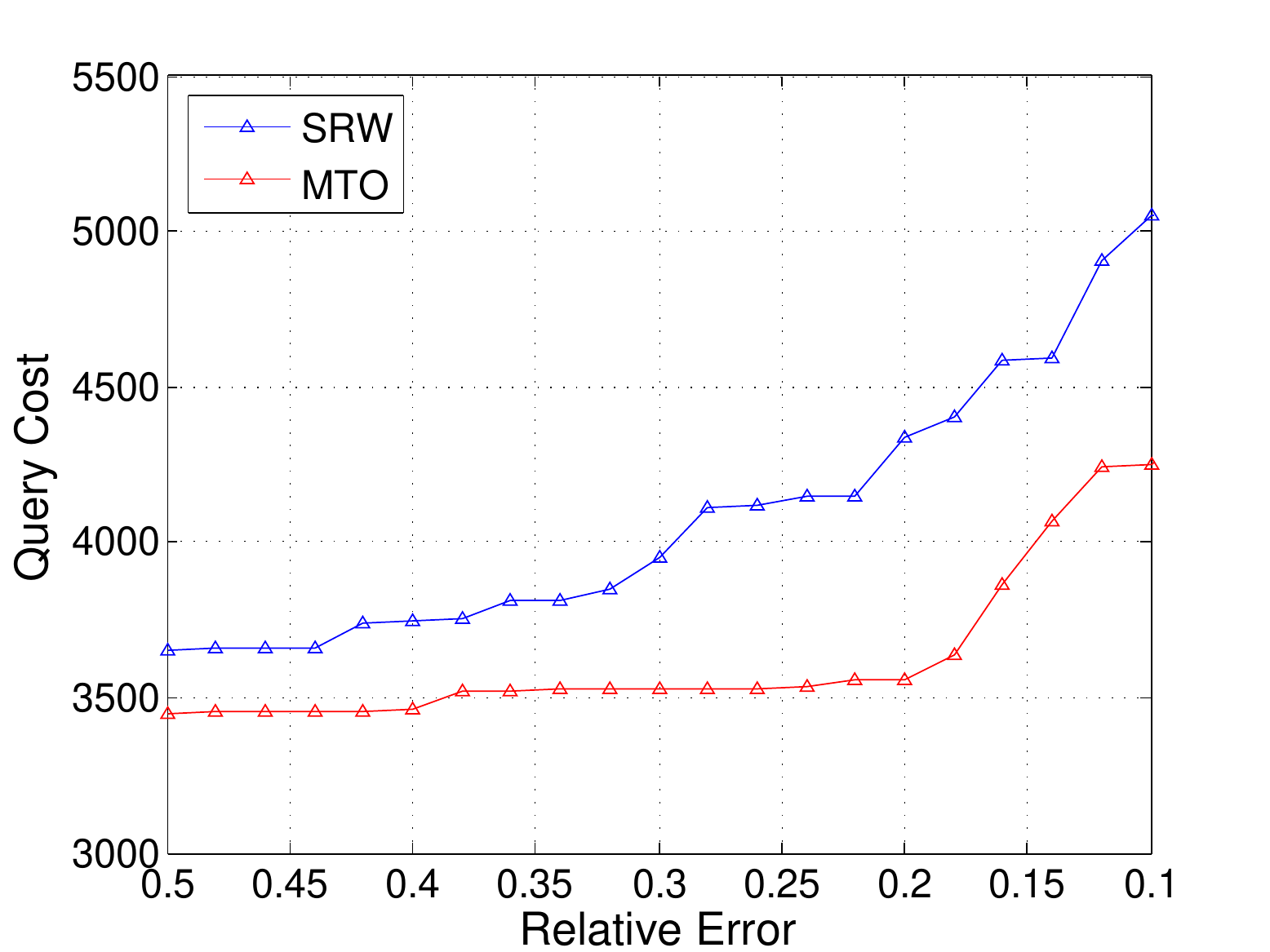}
        }
    \end{center}
    \caption{Google Plus online social network}    
    \label{fig:googleplus}
\end{figure*}

Fig~\ref{fig:gp_degree} shows the estimated average degree when running SRW and MTO-Sampler random walk on Google Plus. It clearly shows that MTO-Sampler's variance is smaller and converges faster than simpler random walk.  Fig~\ref{fig:gp_error_deg} and \ref{fig:gp_introlen} illustrate the comparison between SRW and MTO of the relative error vs query cost of multiple attributes. We note that the self-description length is the number of characters in users' self-description. One can see that our MTO-Sampler significantly outperforms SRW. 

\vspace{2mm}
\noindent{\bf Synthetic Social Networks:}
\label{sec:theoretical_exp}
Finally, we conducted further analysis of our MTO-Sampler, in particular the individual effects of edge removals (RM) and edge replacements (RP), using the synthetic latent space model described in Section~\ref{sec:epd}. Fig~\ref{graph_eig} depicts the results when the number of nodes in the graph varies from 50 to 100 (with the latent space model, we distributed these nodes in an area of $[0,4]\times [0,5]$, and set $r=0.7$). We derived the theoretical mixing time from the second largest eigenvalue modulus of the transition matrix. Note that Fig~\ref{graph_eig} also includes the theoretical bound derived in Section 4.2. One can see from the figure that our final MTO-Sampler achieves better efficiency than the individual applications of edge removal and replacement. In addition, the theoretical model represents a conservation estimation that is outperformed by the real efficiency of MTO-Sampler - consistent with our results in Section 4.2.


\section{Related Work}
\noindent\textbf{Sampling from online social networks.} Several papers \cite{Leskovec2006a, Airoldi:2005:SAP:1117454.1117457, Kurant} have considered sampling from general large graph, and \cite{Katzir2011, Mohaisena, Gjoka2010} focus on sampling from online social networks. 

With global topology, \cite{Leskovec2006a} discussed sampling techniques like random node, random edge, random subgraph in large graphs. \cite{Jin} introduced Albatross sampling which combines random jump and MHRW. \cite{Gjoka2010} also demonstrated true uniform sampling method among the users' id as ``ground-truth". 

Without global topology, \cite{Gjoka2010, Leskovec2006a}  compared sampling techniques such as Simple Random Walk, Metropolis-Hastings Random Walk and traditional Breadth First Search (BFS) and Depth First Search (DFS). Also \cite{Gjoka2010, Alon2008} considered many parallel random walks at the same time, and MTO-sampler can be applied to each parallel random walk straightforwardly, since it is an parameter-free and online algorithm. 

Moreover, to the best of our knowledge, random walk is still the most practical way to sampling from large graphs without global topology.

\noindent\textbf{Shorten the mixing time of random walks.} \cite{Mohaisena} found that the mixing time of typical online social networks is much larger than anticipated, which validates our motivation to shorten the mixing time of random walk. \cite{Boyd2004} derived the fastest mixing random walk on a graph by convex optimization on second largest eigenvalue of the transition matrix, but it need the whole topology of the graph,  and its high time complexity make it inapplicable in large graphs. 

\noindent\textbf{Theoretical models of online social network.} \cite{Sarkar:2005:DSN:1117454.1117459} compared latent space model with real social network data.
\cite{Chung04thesmall} introduced hybrid graph model to incorporate the small world phenomenon. \cite{Sala2010a} also measured the difference between multiple synthetic graphs and real world social network graphs.


\section{Conclusions}

In this paper we have initiated a study of enabling faster random walk over an online social network (with a restrictive web interface) by ``rewiring'' the social network on-the-fly. We showed that the key for speeding up a random walk is to increase the conductance of the graph topology followed by the random walk. As such, we developed MTO-Sampler which provably increases the graph conductance by constructing an overlay topology on-the-fly through edge removals and replacements. We provided theoretical analysis and extensive experimental studies over real-world social networks to illustrate the superiority of MTO-Sampler on achieving a smaller sampling bias while consuming a lower query cost.

\balance


\bibliographystyle{abbrv}
\bibliography{12_VLDB} 

\begin{thebibliography}{10}

\bibitem{stanford_dataset}
Stanford large network dataset collection \url{http://snap.stanford.edu/data/}.

\bibitem{Airoldi:2005:SAP:1117454.1117457}
E.~M. Airoldi.
\newblock Sampling algorithms for pure network topologies.
\newblock {\em SIGKDD Explorations}, 7:13--22, 2005.

\bibitem{Alon1986}
N.~Alon.
\newblock Eigenvalues and expanders.
\newblock {\em Combinatorica}, 6:83--96, 1986.
\newblock 10.1007/BF02579166.

\bibitem{Alon2008}
N.~Alon, C.~Avin, M.~Koucky, G.~Kozma, Z.~Lotker, and M.~R. Tuttle.
\newblock Many random walks are faster than one.
\newblock In {\em SPAA}, 2008.

\bibitem{Boyd2004}
S.~Boyd, P.~Diaconis, and L.~Xiao.
\newblock Fastest mixing markov chain on a graph.
\newblock {\em SIAM REVIEW}, 46:667--689, 2003.

\bibitem{Boyd2005}
S.~Boyd, A.~Ghosh, and B.~Prabhakar.
\newblock {Mixing times for random walks on geometric random graphs}.
\newblock {\em SIAM ANALCO}, 2005.

\bibitem{Chung2007b}
F.~Chung.
\newblock {Random walks and local cuts in graphs}.
\newblock {\em Linear Algebra and its Applications}, 423(1):22--32, May 2007.

\bibitem{Chung04thesmall}
F.~Chung and L.~Lu.
\newblock The small world phenomenon in hybrid power law graphs.
\newblock In {\em Complex Networks, (Eds. E. Ben-Naim et. al.),
  Springer-Verlag}, pages 91--106. Springer, 2004.

\bibitem{Geweke92evaluatingthe}
J.~Geweke.
\newblock Evaluating the accuracy of sampling-based approaches to the
  calculation of posterior moments.
\newblock In {\em IN BAYESIAN STATISTICS}, pages 169--193. University Press,
  1992.

\bibitem{Gjoka2010}
M.~Gjoka, M.~Kurant, C.~T. Butts, and A.~Markopoulou.
\newblock Walking in facebook: A case study of unbiased sampling of osns.
\newblock In {\em INFOCOM}, 2010.

\bibitem{Jin}
L.~Jin, Y.~Chen, P.~Hui, C.~Ding, T.~Wang, A.~V. Vasilakos, B.~Deng, and X.~Li.
\newblock Albatross sampling: robust and effective hybrid vertex sampling for
  social graphs.
\newblock In {\em MobiArch}, 2011.

\bibitem{Katzir2011}
L.~Katzir, E.~Liberty, and O.~Somekh.
\newblock {Estimating sizes of social networks via biased sampling}.
\newblock In {\em WWW}, 2011.

\bibitem{Kurant}
M.~Kurant, M.~Gjoka, C.~T. Butts, and A.~Markopoulou.
\newblock Walking on a graph with a magnifying glass: stratified sampling via
  weighted random walks.
\newblock In {\em SIGMETRICS}, 2011.

\bibitem{Lee2012}
C.-H. Lee, X.~Xu, and D.~Y. Eun.
\newblock Beyond random walk and metropolis-hastings samplers: Why you should
  not backtrack for unbiased graph sampling.
\newblock In {\em Sigmetrics}, 2012.

\bibitem{Leskovec2006a}
J.~Leskovec and C.~Faloutsos.
\newblock Sampling from large graphs.
\newblock In {\em SIGKDD}, 2006.

\bibitem{Leskovec2009}
J.~Leskovec, K.~Lang, A.~Dasgupta, and M.~Mahoney.
\newblock {Community structure in large networks: Natural cluster sizes and the
  absence of large well-defined clusters}.
\newblock {\em Internet Mathematics}, 6(1):29--123, 2009.

\bibitem{Lovasz1993}
L.~Lov\'{a}sz.
\newblock {Random walks on graphs: A survey}.
\newblock {\em Combinatorics, Paul Erdos is Eighty}, 2(1):1--46, 1993.

\bibitem{Mohaisena}
A.~Mohaisen, A.~Yun, and Y.~Kim.
\newblock {Measuring the mixing time of social graphs}.
\newblock In {\em SIGCOMM}, 2010.

\bibitem{Richardson2003}
M.~Richardson, R.~Agrawal, and P.~Domingos.
\newblock Trust management for the semantic web.
\newblock In {\em ISWC}, 2003.

\bibitem{Sala2010a}
A.~Sala, L.~Cao, C.~Wilson, R.~Zablit, H.~Zheng, and B.~Zhao.
\newblock {Measurement-calibrated graph models for social network experiments}.
\newblock In {\em WWW}, 2010.

\bibitem{Sarkar_theoreticaljustification}
P.~Sarkar, D.~Chakrabarti, and A.~W. Moore.
\newblock Theoretical justification of popular link prediction heuristics.
\newblock In {\em COLT}, 2010.

\bibitem{Sarkar:2005:DSN:1117454.1117459}
P.~Sarkar and A.~W. Moore.
\newblock Dynamic social network analysis using latent space models.
\newblock {\em SIGKDD Explor. Newsl.}, 7:31--40, December 2005.

\bibitem{technical_report}
Z.~Zhou, N.~Zhang, Z.~Gong, and G.~Das.
\newblock \url{http://www.seas.gwu.edu/~nzhang10/rewiring.pdf}.

\end{thebibliography}

\section{Appendix}

\begin{figure}
    \centering
    \includegraphics[height=1.6in]{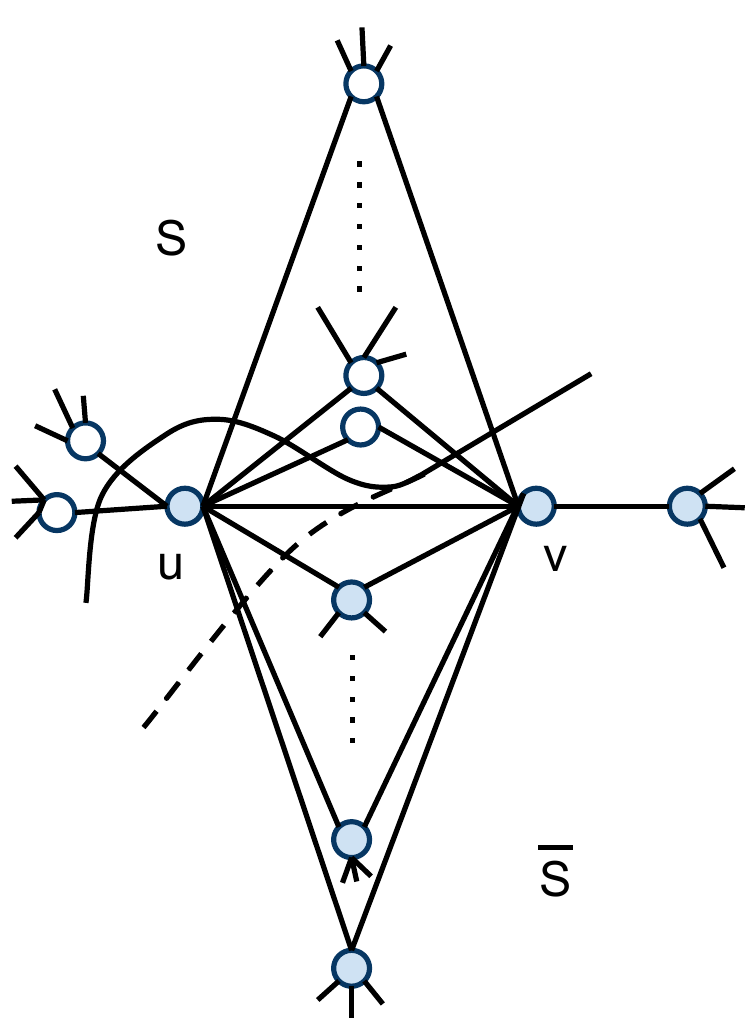}
    \caption{An counter example of lemma \ref{delete_propo2} that edge $e_{uv}$ must be a cross-cutting edge.}
    \label{delete_tightness}
\end{figure}

\vspace{5mm}{\bf Corollary 1.} {\it For all $N(u), N(v), k_u, k_v$ which satisfy 
\begin{equation}
\left\lceil \frac{|N(u)\cap N(v)|}{2} \right\rceil +1 \leq \frac{1}{2}\max\{k_u,k_v\},
\label{delete_equ2}
\end{equation} 
there always exists a graph $G(V, E)$ in which $e_{uv}$ is cross-cutting. }\vspace{5mm}

\begin{proof} Let $n=|N(u)\cap N(v)|$. We only need to construct a counter-example for each case that satisfies (\ref{delete_equ2}), but $e_{uv}$ is a cross-cutting edge. Assume we have a graph like Fig \ref{delete_tightness}, which shows the whole view of it. We let the number of common neighbors of node $u$ and $v$ be $n$. Assuming $k_u\geq k_v$, from (\ref{delete_equ2}) we get:
    \begin{equation}
    \left\{
    \begin{array}{ll}
        |O_u|=\max\{k_u,k_v\}-n-1\geq 1,&\text{if $n$ is even}\\
        |O_u|=\max\{k_u,k_v\}-n-1\geq 2,&\text{if $n$ is odd}
    \end{array}
    \right.  
    \end{equation}
    Here $O_u = \{e_{wu}|w\in N(u)-N(v)\cup\{v\}\}$, which denotes the outer edges of $u$ which is not linked to the node $v$ and their common neighbors. We can carefully construct a graph like Fig \ref{delete_tightness}: for each neighbor of node $u$ and $v$, it only has 1-degree neighbors. So we need to prove that after assigning the degree for each node, $e_{uv}$ will be a cross-cutting edge. If we simply let:
    \begin{equation}k_w\gg max\{k_u,k_v\} \,\,\,\forall w\in \{V-\{u,v\}\}\end{equation}
and then we divide these nodes into two sets $S$ and $\bar{S}$. 

Suppose $n$ is even. In order to achieve the minimum in the definition of conductance, there must exist the case such that we only need to decide whether node $u$ is in $S$ or in $\bar{S}$. 
    \begin{equation}
    \text{\# Cross-Cutting Edges}=\left\{
    \begin{array}{ll}
        n+1,&\text{if $u\in S$}\\
        n+|O_u|,&\text{if $u\in \bar{S}$}
    \end{array}
    \right.
    \end{equation}
If $|O_u|>1$, we can easily assert that $e_{uv}$ is a cross-cutting edge. If $|O_u|=1$, we can let $|a(S)|=|a(\bar{S})|$ when $u\in S$ to minimize $\min\{|a(S)|,|a(\bar{S})|\}$. So $e_{uv}$ is an cross-cutting edge under this circumstance.

Also, suppose $n$ is odd. Similarly, we have
    \begin{equation}
    \text{\# Cross-Cutting Edges}=\left\{
    \begin{array}{ll}
        n+1,&\text{if $u\in S$}\\
        2\lfloor\frac{n}{2}\rfloor+|O_u|,&\text{if $u\in \bar{S}$}
    \end{array}
    \right.
    \end{equation}
Since $|O_u|\geq 2$, in the same way we know that $e_{uv}$ is a cross-cutting edge.
\end{proof}

\begin{figure}
    \centering
    \includegraphics[width=1.6in]{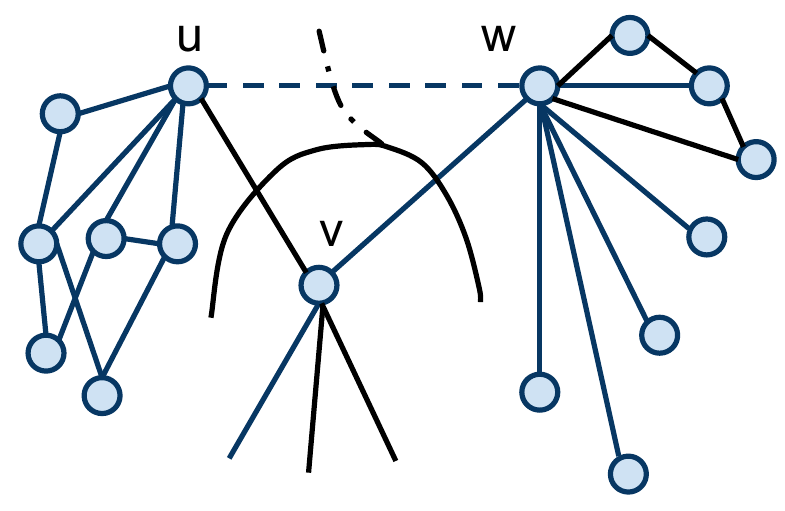}
    \caption{When replacing $e_{uv}$ with $e_{uw}$ from a node $v$ whose degree is larger than 3, it may decrease the conductance if both $e_{uv}$ and $e_{wv}$ are cross-cutting edges.}
    \label{replace_tightness}
\end{figure}

\vspace{5mm}{\bf Theorem 4.} {\it Given $G(V,E)$, $\forall v\in V$, if $k_v=3$, $u,w \in N(v)$, then replacing edge $e_{uv}$ with $e_{uw}$ will not decrease the conductance, while it also has positive possibility to increase the conductance. 
\label{replace_propo}}\vspace{5mm}
\begin{proof} 
First, no matter $e_{uv}$ is a cross-cutting edge or not, replace it with $e_{uw}$ should at least obtain the same conductance. If $e_{uv}$ is not a cross-cutting edge, then obviously we are not going to decrease the conductance because $a(S)$ or $a(\bar{S})$ will not change. If $e_{uv}$ is a cross-cutting edge, we only need to prove that $e_{uw}$ is also a cross-cutting edge. Let's assume $e_{uw}$ is not a cross-cutting edge, then we can infer that the $e_{vw}$ is a cross-cutting edge. But $v$ only has degree of 3, so it is obvious that letting $u$, $v$ and $w$ be the same side will achieve less conductance, which contradicts the definition of conductance. 

But if $e_{vw}$ is a cross-cutting edge, and we replace $e_{uv}$ with $e_{uw}$, then $e_{uw}$ has the positive probability to become one more cross-cutting edge in this local view of $u$, $v$ and $w$, which result in higher conductance.
\end{proof}

\vspace{5mm}{\bf Corollary 2.} {\it \label{replace_tight} For $v\in V$, if $k_v \neq 3$, then there always exist a graph 
$G(V,E)$, $\forall u,w \in N(v)$, such that replacing $e_{uv}$ with $e_{uw}$ will decrease the conductance or have no effect.}\vspace{5mm}
\begin{proof}
If $k_v=1$, then we could not cut it to disconnect the graph.
If $k_v=2$, we need to check some possible situations. If none of these edges linked to $v$ are cross-cutting edges, then replacing would not has effect on the conductance. If either $e_{uv}$ or $e_{wv}$ is a cross-cutting edge, then replace one of them with $e_{uw}$ will not generate another cross-cutting edge; because now $k_v=1$, and it should belongs to one side of the separation, $S$ or $\bar{S}$. 

So we only need to consider the situation when $k_v\geq 4$. See Fig \ref{replace_tightness}. There exist the case when both $e_{uv}$ and $e_{wv}$ are cross-cutting edges. Then replacing $e_{uv}$ with $e_{wv}$ would decrease the number of cross-cutting edges from 2 to 1 locally, which may lead to dramatic decrease of the conductance of the graph.

The uniqueness of $k_v=3$ is that there would not exist the case when both $e_{uv}$ and $e_{wv}$ are cross-cutting edges.
\end{proof}

\vspace{5mm}{\bf Theorem 5.} {\it Given $G(V, E)$, $\forall u, v \in V$, if $e_{uv}\in E$ and
\begin{equation}
\left\lceil\frac{|N(u)\cap N(v)|-N^*}{2}\right\rceil + 1 + \frac{1}{2}\sum_{w\in N^*}(4-k_w) > \frac{1}{2}\max\{k_u,k_v\},
\label{delete_equ3}
\end{equation}
we can assert that $e_{uv}$ is not a cross-cutting edge. Here we denote $N^* = \{w\in N(u)\cap N(v) | \,\,k_w\text{ is known }, 2\leq k_w\leq 3 \}$.
\label{delete_propo3}}\vspace{5mm}
\begin{proof}
Noticed that if we do not know any degree information about the common neighbors of $u$ and $v$, then $N^*=\emptyset$, and theorem \ref{delete_propo3} is exactly the same as theorem \ref{delete_propo1}. 

We are going to prove this theorem by contradiction, which means if we assume $e_{uv}$ is a cross-cutting edge, then we can find another configuration of $S$ and $\bar{S}$ such that $e_{uv}$ is not a cross-cutting edge but obtain less conductance. Again, let $n=|N(u)\cap N(v)|$, according to the assumption the number of common neighbors of $u$ and $v$ is n, then there must be $n+1$ cross-cutting edges in this local view of the graph, see Fig \ref{delete_extra}. 

Given a node $w\in N(u)\cap N(v)$, and according to some historical information we can achieve its degree $k_w$ without paying any query cost. So obviously, if $k_w\geq 4$ then it makes no sense to consider the rearrangement of it because dragging $w$ from $S$ to $\bar{S}$ would probably increase the number of cross-cutting edges without knowing the edge information outside this local view of the graph. Therefore, we only need to consider $N^*$, which is the set of all the nodes belongs to common neighbors of degree 2 and 3. 

if we denote that the number of cross-cutting edges linked to $u$ within $N^*\cup \{u\}$ is $n_u^{(i)}$, the number of cross-cutting edges linked to u outside $N^*\cup \{u\}$ is $n_u^{(o)}$, and similarly we have $n_v^{(i)}$ and $n_v^{(o)}$. So we have $n_u^{(i)}+n_u^{(o)}+n_v^{(i)}+n_v^{(o)}=n$. According to the condition described in Proposition \ref{delete_propo3}, either of the following inequality would hold:
\begin{align*}
n_u^{(o)}+n_v^{(i)} + 1 &> \frac{1}{2}\left(\max\{k_u,k_v\}\right) 
\\ &-\frac{1}{2}\left( 2|N^*|- \sum_{w\in N^*}(k_w-2)\right)\\
n_u^{(i)}+ n_v^{(o)} + 1 &> \frac{1}{2}\left(\max\{k_u,k_v\}\right) \\&-\frac{1}{2}\left( 2|N^*|- \sum_{w\in N^*}(k_w-2)\right)
\end{align*}
Without losing generality, assume the first one holds, then we are going to prove that by rearrange the set $N^*\cup\{u\}$ we can achieve a lower conductance and thus lead to the contradiction. 

Imagine that if we try to drag the whole set of $N^*\cup\{u\}$ from $S$ to $\bar{S}$, then we need to ``rearrange" all the edges linked to the set: those cross-cutting edges linked to the set but outside $N^*$ will be ``fliped", i.e. from cross-cutting edges to non-cross-cutting edges and vice versa; those cross-cutting edges linked to the set but inside $N^*$ will be eliminated, otherwise there will be two cross-cutting edges linked to the node in $N^*$, which is impossible because $\forall w\in N^*$, $2\leq k_w\leq 3$.

Let $O_{N^*\cup\{u\}} = \{e_{vw}\in E|v\in (N^*\cup\{u\}), w\notin (N^*\cup\{u\})\}$, then
\begin{equation}|O_{N^*\cup\{u\}}|=\max\{k_u,k_v\} + \sum_{w\in N^*}k_w - 2|N^*|.\end{equation}
And we know that the minimum number of cross-cutting edges we can manipulate will be at least $\left\lceil\frac{n-N^*}{2}\right\rceil+1+ N^*$. So as the result of one line calculation of (\ref{delete_equ3}), \begin{equation}\left\lceil\frac{n-N^*}{2}\right\rceil+1+ N^*>\frac{1}{2}|O_{N^*\cup\{u\}}|. \end{equation} Therefore moving the set $N^*\cup\{u\}$ from $S$ to $\bar{S}$ will always results in a lower conductance.
\end{proof}
\begin{figure}
    \centering
    \includegraphics[width=1.6in]{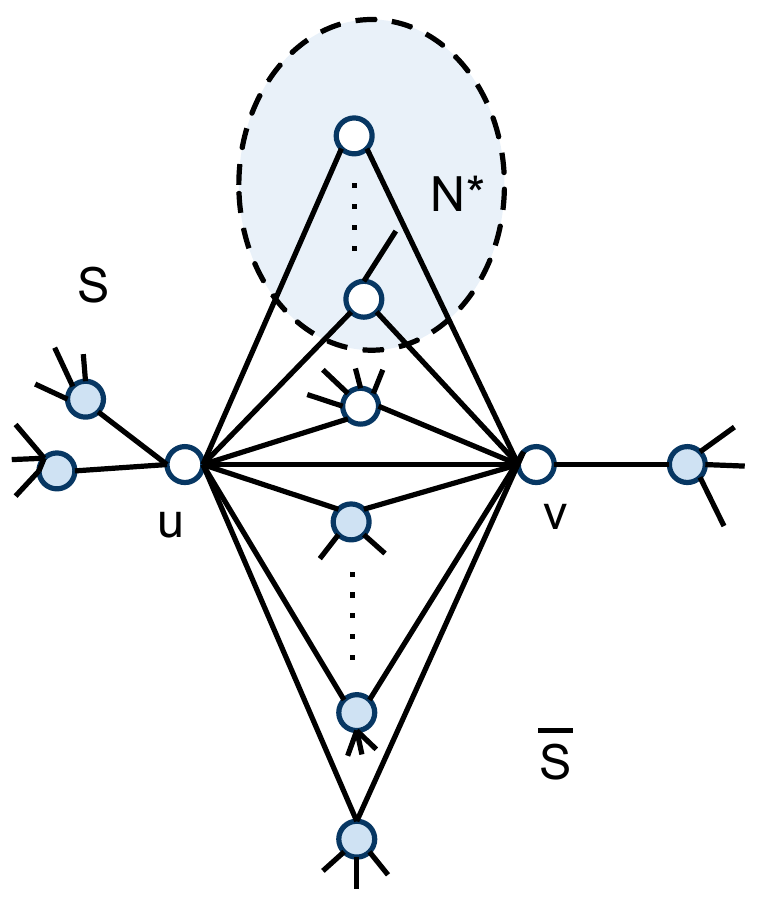}
    \caption{$N^*$ is the set we have accessed before and whose nodes are of degree 2 and 3. We do not know the blue nodes' degree.}
    \label{delete_extra}
\end{figure}

\vspace{5mm}{\bf Theorem 6.} {\it Given a latent space graph model G(V,E), assume $\alpha=+\infty$, then the expected number of edges we can removed
\begin{equation}
\mathbb{E}[R] \geq |E|\cdot\mathbb{P}\left(d<V(r)\left(1-\left(\frac{1}{3}\right)^{1/D}\right)\right)
\label{eq:E-R}
\end{equation}
Moreover, if we assume the dimension $D=2$, and nodes are uniformly distributed in a rectangle $[0,a]\times [0,b]$, then for the graph $G^*$ (after removing edges from G) is:
\begin{align}
\mathbb{E}[\Phi(G^*)] \geq \frac{\Phi(G)}{1-\iint_{z_1^2+z_2^2\leq 0.75r^2}f_a(z_1)f_b(z_2)dz_1z_2}
\end{align}
where $z_1$ and $z_2$ are independent uniform random variable supported on $[0,a]$ and $[0,b]$.}\vspace{5mm}

\begin{proof}
According to \cite{Sarkar_theoreticaljustification}, we have
\begin{equation}V(r)\left( 1-\frac{d_{ij}}{2r}\right)^D \leq \frac{|N(i)\cap N(j)|}{|N(i)\cup N(j)|}.\end{equation}
$V(r)$ is the volume of a $D$ dimensional hypersphere of radius r. 
Therefore, if we have small enough $d_{ij}$, than we can confirm that we can remove the edge $e_{ij}$. Conservatively, from theorem \ref{delete_propo1} we can reasonably assert that if $|N(i)\cap N(j)|\geq |N(i)\cup N(j)|-2$, then the edge $e_{ij}$ can be safely removed. So when \begin{equation}d_{ij}\leq 2r\left( 1- \left(\frac{1}{V(r)}(1-\frac{2}{|N(i)\cup N(j)|}) \right)^{1/D} \right)=d_0,\end{equation} the edge $e_{ij}$ can be removed. Now, we have transformed the probability of removing an edge to the probability of two node's distance is within a threshold. Since $|N(i)\cup N(j)|\geq 3$, so (\ref{eq:E-R}) holds. 

Given more assumptions of dimension and the distribution of nodes, the probability of two nodes' euclidean distance smaller than the threshold is:
\begin{equation}
\mathbb{P}(d\leq d_0) = \iint_{z_1^2+z_2^2\leq d_0^2}f_a(z_1)f_b(z_2)dz_1z_2.
\end{equation}
Also, since $|N(u)\cap N(v)| \geq 3$, the change of conductance can be calculated as 
\begin{align}
\mathbb{E}[\Phi(G^*)] &= \frac{|\sigma(S)|}{a(S)-\mathbb{P}(d\leq d_0)a(S)}\\ &= \frac{1}{1-\mathbb{P}(d\leq d_0)}\Phi(G)
\\ &\geq \frac{\Phi(G)}{1-\iint_{z_1^2+z_2^2\leq 0.75r^2}f_a(z_1)f_b(z_2)dz_1z_2}.
\end{align}
\end{proof}

\end{document}